\documentclass[
    onecolumn,10pt,aps,pra,
    superscriptaddress,
    nofootinbib,notitlepage
]{revtex4-2}

\usepackage{graphicx} 
\usepackage[]{xcolor}

\usepackage{amssymb}
\usepackage{amsmath}
\usepackage{amsthm}
\usepackage{mathtools}

\usepackage{bm}
\usepackage{bbold} 
\usepackage{siunitx}
\usepackage{physics}
\usepackage[version=4]{mhchem}
\usepackage{enumitem}
\usepackage{tabularx}

\usepackage{qcircuit}

\definecolor{LEI-blue}{cmyk}{1,.75,0,.35} 
\definecolor{LEI-orange}{cmyk}{0,.62,.97,0} 
\definecolor{niceblue}{rgb}{.1, .25, .8}

\usepackage[colorlinks=true, linkcolor=niceblue, citecolor=gray]{hyperref}

\newcommand{\topic}[1]{
}

\newtheorem{theorem}{Theorem}
\newtheorem{definition}[theorem]{Definition}

\newtheorem{algorithm}[theorem]{Algorithm}

\DeclareMathOperator*{\argmax}{arg\,max}
\DeclareMathOperator*{\argmin}{arg\,min}
\DeclareMathOperator{\Var}{Var}

\DeclareMathOperator{\sgn}{sgn}

\DeclareMathOperator{\E}{\mathbb{E}}

\newcommand{\NN}{\mathbb{N}}

\newcommand{\equalcontrib}{
    \altaffiliation{These authors contributed equally to this work}
}

\newcommand{\lorentz}{
    \affiliation{Lorentz Institute for Theoretical physics, Leiden University, The Netherlands}
}

\newcommand{\google}{
    \affiliation{Google Quantum AI}
}

\newcommand{\covestro}{
    \affiliation{Covestro Deutschland AG, Leverkusen 51373, Germany}
}

\begin{document}

\title{Error mitigation and circuit division for early fault-tolerant quantum phase estimation}
\author{Alicja Dutkiewicz}
\equalcontrib
\email{dutkiewicz@lorentz.leidenuniv.nl}
\lorentz

\author{Stefano Polla}
\equalcontrib
\email{polla@lorentz.leidenuniv.nl}
\lorentz
\google

\author{Maximilian Scheurer}
\covestro

\author{Christian Gogolin}
\covestro
\lorentz

\author{William J. Huggins}
\google

\author{Thomas E. O'Brien}
\email{teobrien@google.com}
\google
\lorentz

\date{\today}

\begin{abstract}
As fully fault-tolerant quantum computers capable of solving useful problems remain a distant goal, we anticipate an era of \emph{early fault tolerance} where limited error correction is available.
We propose a framework for designing early fault-tolerant algorithms by trading between error correction overhead and residual logical noise, and apply it to quantum phase estimation (QPE).
We develop a quantum-Fourier-transform (QFT)-based QPE technique that is robust to global depolarising noise and outperforms the previous state of the art at low and moderate noise rates.
We further introduce the Explicitly Unbiased Maximum Likelihood Estimation (EUMLE), a data processing technique that mitigates \emph{arbitrary} errors in QFT-based QPE schemes. EUMLE provides consistent, asymptotically normal error-mitigated estimates, addressing the open problem of extending error mitigation beyond expectation value estimation.
Applying this scheme to the ground state problem of the two-dimensional Hubbard model and various molecular Hamiltonians, we find we can roughly halve the number of physical qubits with a $\sim 10\times$ wall-clock time overhead, but further reduction causes a steep runtime increase.
This work provides an end-to-end analysis of early fault-tolerance cost reductions and space-time trade-offs,  and identifies which areas can be improved in the future.
\end{abstract}

\maketitle

\tableofcontents

\section{Introduction}

Despite significant experimental progress in recent years~\cite{acharyaSuppressing2023,silvaDemonstration2024,bluvsteinLogical2024,acharyaQuantum2024}, quantum computers still require many orders of magnitude more qubits to achieve advantage for practically relevant computational problems~\cite{reiherElucidating2017,babbushEncoding2018,leeEven2021}.
While quantum error correction (QEC) promises to enable fault-tolerant quantum computing with a cost that scales polylogarithmically in the size of the computation, the overhead required for the most practically realizable schemes is large.
For example, the combination of the \(2\)-dimensional surface code with magic state distillation~\cite{fowlerSurface2012, fowlerLow2019, litinskiGame2019, gidneyEfficient2019} is capable of universal quantum computation, but may require hundreds or thousands of physical qubits per logical qubit to ensure the success of the computation.
This overhead has prompted a push for work on \emph{early fault-tolerance}~\cite{campbellEarly2022,katabarwaEarly2023,zhangComputing2022,kshirsagarProving2024,wangState2022,wangFaster2023,nelsonAssessment2024,wanRandomized2022,linHeisenbergLimited2022,dongGround2022,dingEven2023,ding2023robust,bultriniBattle2023,liangModeling2024,akahoshiPartially2024,toshioPractical2024,akahoshiCompilation2024}, that is, the search for computational frameworks and trade-offs that enable us to take advantage of the machinery of quantum error correction without incurring its full cost.

Initial explorations of early fault-tolerance have focused along a few lines.
A fist line of research has aimed to reduce costs for small and medium-sized problems, by making tailored algorithm design choices. 
For example, using Trotter-based simulation instead of methods with better asymptotic scaling has been shown to offer lower costs for smaller systems~\cite{childsFirst2018,campbellEarly2022,bluntCompilation2024,toshioPractical2024}.
Another approach has focused on adapting quantum algorithms to achieve a low depth per circuit.
This is motivated by the observation that a given error rate~\(p\), determines the maximum number of allowable operations $G\approx-\log(1-p)^{-1} \approx p^{-1}$ before the exponential decay of circuit fidelity sets in.
Several works have accomplished this by trading off reduced circuit depth for an increased number of circuit repetitions~\cite{wangAccelerated2019, linHeisenbergLimited2022, zhangComputing2022, wanRandomized2022, wangFaster2023, dongGround2022, wangFaster2023, dingEven2023, dingSimultaneous2023, ding2023robust, niLowdepth2023, nelsonAssessment2024, akhalwayaTopological2024}.
A third approach has been to use error mitigation techniques -- originally designed to handle physical noise in near-term quantum devices -- in order to reduce the noise-induced bias without reducing circuit depth, albeit at the expense of many circuit repetitions \cite{piveteauError2021,suzukiQuantum2022,toshioPractical2024}.

Quantum error mitigation (EM)~\cite{caiQuantum2023} has seen increased interest in recent years, due to the development of experimental platforms which are increasingly powerful, yet remain far from enabling fault-tolerant computing.
Mitigating errors caused by physical noise avoids the large constant factor requirements of fault tolerance, but results in an exponential overhead with increasing circuit depth and error rate~\cite{aharonovLimitations1996,aharonovPolynomialTime2023,takagiFundamental2022,quekExponentially2023,takagiUniversal2023, tsubouchi2023universal}.
A wide range of EM techniques have been conceived~\cite{temmeError2017,endoPractical2018, liEfficient2017,mcardleErrormitigatedDigitalQuantum2019,bonet-monroigLowcostErrorMitigation2018,hugginsVirtual2021,hugginsVirtual2021,obrienError2021,montanaroError2021,caiQuantum2023} and implemented in practice~\cite{kandalaError2019, sagastizabalExperimental2019, vandenbergProbabilistic2023, stanisicObserving2022, obrienPurificationbased2023, aruteHartreeFock2020}, demonstrating up to thousandfold-reductions in noise bias~\cite{obrienPurificationbased2023}.
While these methods cannot replace scalable error correction, they have been suggested as a last mile solution in fault-tolerant quantum algorithms, relaxing the demands on a QEC protocol by enabling a computation to tolerate \(\mathcal{O}(1)\) errors per circuit repetition~\cite{piveteauError2021,suzukiQuantum2022,katabarwaEarly2023,akahoshiPartially2024,caiQuantum2023}.
Some approaches specifically address scenarios where the error rates vary between qubits or operations~\cite{piveteauError2021,suzukiQuantum2022,bultriniBattle2023,toshioPractical2024}, with a particular emphasis on the non-Clifford operations that are expected to be the most costly to implement fault-tolerantly.
However, most EM protocols have focused on expectation value estimation, as opposed to more general quantum computation; extending beyond this was labeled an outstanding challenge in Ref.~\cite{caiQuantum2023}.

Quantum phase estimation~\cite{kitaevQuantumMeasurementsAbelian1995}, a workhorse of quantum computing, aims to compute an eigenvalue of a unitary matrix given its circuit implementation and an initial state with sufficiently high overlap with the target eigenstate.
Traditional phase estimation techniques~\cite{nielsen2001quantum} use a control register which, after accumulating phase information via controlled time-evolution, undergoes the quantum Fourier transform (QFT) to yield a state with high density on binary representations of the desired phase.
These \emph{QFT-based QPE} algorithms achieve the Heisenberg limit, the information-theoretic optimum for phase estimation, including constant factors~\cite{vandamOptimal2007}.
It is possible to remove the quantum Fourier transform, either by direct dequantization~\cite{kitaevQuantumMeasurementsAbelian1995,najafiOptimum2023}, or by analysing the outputs of independent circuits that use a single control qubit via signal processing techniques~\cite{kimmelRobust2015,obrien2019quantum,sommaQuantum2019,linHeisenbergLimited2022,dutkiewicz2022heisenberg,dongGround2022,dingSimultaneous2023,wanRandomized2022}.
(Indeed, under some constraints on $U$ phase estimation without any control qubits is possible~\cite{luAlgorithms2021,obrienError2021,guNoiseresilient2022,yang2024phase}.)
These \emph{single-control QPE} techniques are of particular interest in early-fault tolerance~\cite{linHeisenbergLimited2022,dongGround2022,dingSimultaneous2023,katabarwaEarly2023,liangModeling2024, zhangComputing2022, kshirsagarProving2024, toshioPractical2024, nelsonAssessment2024, kiss2024early}, as these protocols can tolerate a constant amount of error per circuit~\cite{kimmelRobust2015,obrien2019quantum}, avoiding the need for a high overall success probability across all circuit runs.
However, it remains to understand how this reduction in circuit depth and increase in error tolerance propagates through quantum error correction to a trade-off between physical qubit count and wall-clock time (of the full estimation protocol).
Additionally, besides Refs.~\cite{rendon2023low, castaldoDifferentiable2024}, there has been little exploration of whether QFT-based QPE techniques can be divided into smaller circuits akin to single-control methods.
To the best of our knowledge, no prior work applied circuit-division to make QFT-based QPE robust to noise.

In this work, we develop a framework that jointly optimizes quantum error correction, mitigation, and algorithm design for early fault-tolerant quantum computation. 
Specifically, we study the trade-off between running fewer, larger circuits at lower error rates (requiring more physical qubits per logical qubit) versus subdividing computations into many smaller circuits that can be executed at reduced code distances.
Our work bridges two research directions that have so far been explored independently: the optimization of algorithms for early fault-tolerance and the integration of quantum error mitigation within fault-tolerant protocols.
Within this framework, we propose a new noise-robust QPE algorithm, based on subdividing the QFT-based QPE circuit of Ref.~\cite{vandamOptimal2007} used in state-of-the-art fault-tolerant benchmarks \cite{babbushEncoding2018,leeEven2021}. 
In the presence of global depolarizing noise, we prove that the algorithm converges to any target precision and quantify its cost as a function of noise strength.
We compare our method to the noise-optimized robust phase estimation (RPE) scheme of Ref.~\cite{belliardoAchieving2020}, finding that our approach requires between 30\% and 130\% of the resources needed for RPE, depending on the target precision and noise strength.
Then, in order to mitigate biases from more general noise models,
we develop a regularised explicitly-unbiased maximum-likelihood estimator, representing the first extension of error mitigation techniques beyond expectation value estimation.
This estimator, when combined with our QPE algorithm, guarantees convergence under arbitrary noise channels. 
For natural noise models based on local Pauli errors, we upper-bound the overhead of our method relative to the global-depolarizing noise case by a factor of $7\times$.
Finally, we present concrete and comprehensive estimates of the the physical qubit count and wall-clock time for computing ground state energies of the two dimensional Hubbard model \cite{babbushEncoding2018}, and of various molecular Hamiltonians compressed with tensor hypercontraction \cite{leeEven2021}.
Our estimates assume surface-code computing \cite{fowlerLow2019} with CCZ factories \cite{gidneyEfficient2019}, 
and quantify how the previously discussed method trade increased computation time for reduced qubit counts.
We find that the available gains are limited by the the error mitigation overhead, along with the large costs associated to magic state distillation; a $2\times$ reduction in qubit count comes with an increase in the total wall time of $2-100\times$.

\section{Summary of results}

\begin{figure}[t]
    \centering
    \includegraphics[width=\textwidth]{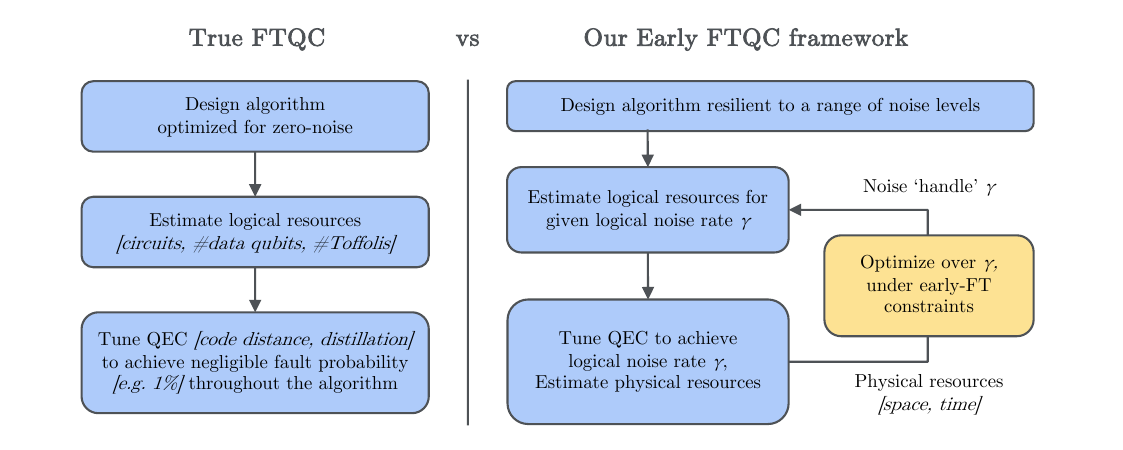}
    \caption{
    Illustration of our framework for designing algorithms for early fault-tolerance and how it differs from traditional fault-tolerant algorithm design and resource estimation. 
    The standard approach (left) involves optimizing an algorithm in the absence of noise and then choosing quantum error correction parameters such that the failure probability is small enough.  
    We propose an alternative approach better suited to early fault-tolerance. 
    By using noise resilient algorithms as a building block, we can allow for non-negligible levels of noise throughout the execution of a circuit, jointly optimizing the parameters of the algorithm, the quantum error correcting scheme, and the error mitigation technique used to address the residual error.}
    \label{fig:design_principles}
\end{figure}

Our overall goal is to extend recent resource estimates for practical quantum phase estimation \cite{babbushEncoding2018,leeEven2021} to the case of early-fault-tolerant quantum computers.
We define early-fault tolerance as \emph{the operational regime where we can use quantum error correction to correct all but \(\Omega(1)\) errors, and we must design algorithms to be robust to the remainder (via e.g. quantum error mitigation)}.
In this setup, we retain the ability to spend more physical resources to reduce the logical error rate $\gamma$ through quantum error correction.
In contrast to standard fault-tolerant quantum algorithm design, we consider $\gamma$ as an additional free parameter to optimize jointly with other parameters of the algorithm, subject to any constraint imposed by the device (e.g.~number of available physical qubits).
The comparison of the standard fault-tolerant algorithm design process with our early fault-tolerant framework is illustrated in Fig.~\ref{fig:design_principles}.
We choose this \emph{top-down} approach to early fault-tolerance (which is similar to the approach taken in Refs.~\cite{katabarwaEarly2023}) instead of a \emph{bottom-up} approach taken by some other works~\cite{akahoshiPartially2024,akahoshiCompilation2024} as this requires relatively few assumptions on the physical device used.

\textbf{Circuit division and error mitigation.} --- 
In order to adapt our quantum algorithm to a range of noise rates~$\gamma$,
we explore \emph{circuit division} protocols that replace a single coherent execution of a quantum circuit consisting of \(G\) gates by the execution of \(M\) shorter circuits, each consisting of less than \(G/R\) gates (with \(M > R\)).
The samples from these shorter circuits are then combined by a classical algorithm to recover the same result as the original quantum algorithm.
Circuit division allows us to trade between executing a few large circuits at lower error rates $\gamma$ (thus requiring more physical qubits per logical qubit), versus dividing these into many shorter circuits (which can tolerate higher $\gamma$ to achieve the same success probability).
Reducing the number of physical qubits required would enable a small early-fault-tolerant quantum computer to perform calculations that would otherwise be out of reach, at the cost of increasing the total run time of the calculation by an $\widetilde{\mathcal{O}}(M/R)$ factor.
Note that this flavor of circuit division is different than divide-and-conquer approaches which have been proposed for e.g.~optimization~\cite{dunjkoComputational2018,geHybrid2020, childsQuantum2022} and variational algorithms~\cite{fujiiDeep2022,perezsalinas2023}.
These are not known to be noise-robust, and the overhead in most cases is exponential in the number of cuts.
For a circuit division protocol to be practical, its overhead $M/R$ needs to be reasonable. 
Furthermore, the classical post-processing should be robust to faults in multiple sub-circuits, as a union bound implies the probability of failure of at least one sub-circuit is larger than the probability of the original circuit failing.

Error mitigation can further be employed to boost the robustness of circuit-division protocols to the remaining sub-circuit faults.
The overhead of EM typically depends on the sub-circuit fidelity $F$ \cite{caiQuantum2023}, which can be controlled either by limiting the sub-circuit depth or by running more error correction.
We provide a general, back-of-the-envelop study of the trade-off between error mitigation and error correction in App.~\ref{sec:early_ft_algs}; we find that combining the two can lead to a significant boost of the maximum logical circuit size that can be run on a given device, in a narrow regime of physical qubit count and error rate.

Single-control QPE variants are readily adaptable to circuit division.
In the presence of global depolarizing noise, they are known to be robust to faults in the sub-circuits with a quadratic overhead $M\sim R^2$~\cite{obrien2019quantum,ding2023robust,katabarwaEarly2023}. 
As these methods rely on measuring expectation values of Hadamard-test-like circuits, they are naturally compatible with existing EM protocols.
This has motivated a lot of interest in recent years~\cite{linHeisenbergLimited2022,guNoiseresilient2022,dingEven2023,wangQuantum2023,dongGround2022,wangFaster2023,dutkiewicz2022heisenberg}, however little work has gone into testing the robustness of more traditional versions of QPE (that utilize a multi-qubit control register and the quantum Fourier transform).

\textbf{QFT-based QPE algorithm robust to depolarizing noise.} --- 
Our first contribution is to provide a circuit-division variant of QFT-based QPE, \emph{Maximum-likelihood Sin-state QPE} (MSQPE), that is robust to global depolarizing noise.
We give the algorithm 
in App.~\ref{sec:sin_state_qpe_multi_circuit}, and prove the following result.
\begin{theorem}[Thm.~\ref{thm:mle_sinqpe}, informal]\label{thm:mle_sinqe_informal} Given an initial eigenstate preparation, the maximum-likelihood sin-state QPE algorithm in the presence of global depolarizing noise with rate $\gamma$ (per call of the unitary) converges to error $\epsilon$ in a total number of uses of the unitary given by
\begin{itemize}
    \item $\pi\epsilon^{-1}$, if $\epsilon\gg \gamma$,
    \item $C\gamma \epsilon^{-2}$ with $C\approx 20$, if $\epsilon \ll \gamma$,
\end{itemize}
and interpolates between these limits when $\epsilon\sim\gamma$.    
\end{theorem}
This algorithm relies on two technical developments: the extension of QFT-based QPE methods to allow for an arbitrary number of calls to the unitary (as opposed to powers of 2 only), and a noise-robust classical post-processing of multiple QPE subroutines via maximum-likelihood estimation.

Thm.~\ref{thm:mle_sinqe_informal} demonstrates that QFT-based QPE methods can be made robust to noise, thus allowing for fair comparison with single-control QPE variants.
We optimize the hyperparameters of the (QFT-based) MSQPE algorithm (details in App.~\ref{sec:sin_state_qpe_multi_circuit}), and compare it against the (single-control) RPE algorithm~\cite{kimmelRobust2015} with optimized parameters from Ref.~\cite{belliardoAchieving2020}
In Fig.~\ref{fig:MSQPE_vs_RPE}, we plot the ratio between the costs of the two algorithms to achieve the same target error $\epsilon$, where the cost $\mathcal{T}_{\mathrm{tot}}$ is measured as the total number of unitaries across all circuits used.
We observe both methods perform similarly; MSQPE performs better when the noise rate \(\gamma\) is low compared to the target error \(\epsilon\), by a factor that increases for smaller \(\gamma\).
Vice versa, when the target error is small \(\epsilon \ll \gamma\), the cost of RPE is about 30\% smaller.
In App.~\ref{sec:qpe_gdn_numerical_comparison} we expand on this comparison, and note that both estimators can likely be improved; the MSQPE estimator by optimizing the sin-state in the presence of noise, and the RPE estimator by numerically optimizing the RPE hyperparameters (the optimization in Ref.~\cite{belliardoAchieving2020} aimed to achieve analytic bounds).
However, across the range of errors and noise rates considered, the performance of both estimators falls within a factor $10\times$ of information-theoretic bounds, which places a limit on any further gains.

\begin{figure}
    \centering
    \includegraphics{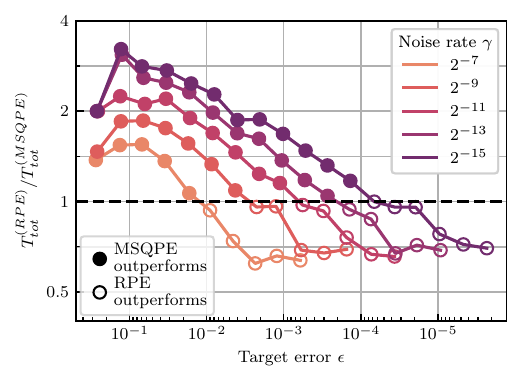}
    \caption{Ratio of the total executions times $\mathcal{T}_{\mathrm{tot}}$ of the optimized robust phase estimation (RPE) algorithm of Ref.~\cite{belliardoAchieving2020} and the maximum-likelihood sin-state quantum phase estimation (MSQPE) algorithm developed in this work across a range of target errors on the phase estimate $\epsilon$ and noise rates $\gamma$ translating to an error rate per unitary of $p_{\mathrm{err}}=(1-e^{-\gamma})$. The dashed line at $1$ denotes equivalent performance between the two methods; points above $1$ give a region where the MSQPE estimator performs better, points below $1$ show where RPE outperforms MSQPE.}
    \label{fig:MSQPE_vs_RPE}
\end{figure}

\textbf{Error mitigation overhead for single-control QPE.} ---
While the global depolarizing channel is a useful analytical tool, it fails to accurately describe the residual logical noise in fault-tolerant quantum computers.
Previous work has shown that other natural noise models (in particular, local depolarizing noise) can result in a biased signal in single-control QPE methods~\cite{obrien2019quantum}, and thus a biased phase estimate.
Consequently, more complex EM strategies will be required to address the general noise models expected to characterize residual logical error.
EM introduces a circuit-repetition overhead that scales inverse-polynomially with the circuit fidelity $F$.

Our second contribution is to study how overhead associated with error mitigation (defined in terms of circuit repetitions) translates to an overhead in the cost of error-mitigated QPE (in terms of number of calls to a noisy implementation of the unitary operator). 
We find that this overhead can be suppressed logarithmically thanks to circuit division.
\begin{theorem}[Thm.~\ref{thm:em_overhead_qpe}, informal]
\label{thm:em_overhead_qpe_informal}
    Given an error mitigation scheme for expectation values with a sample overhead scaling with a power of the circuit fidelity $F^{-\alpha}$, one can obtain an unbiased error mitigation scheme for single-control phase estimation with an overhead proportional to $\alpha \, \gamma/\epsilon$, where $\epsilon$ is the precision of the estimate and $\gamma$ is the noise rate per call of the unitary.
\end{theorem}
This result highlights a significant advantage, as it narrows the cost gap between EM methods with limited scope but low overhead (e.g., postselection) and more general yet costlier approaches (e.g., probabilistic error cancellation).
In Table.~\ref{tab:em_overhead}, we summarize the circuit-repetition overheads $C_\text{em}$ known for various methods, and we show how these translate to unitary-call overheads for phase estimation $C_\text{em}^\text{QPE}$.

\setlength{\tabcolsep}{12pt}
\begin{table}[h]
    \centering
    \begin{tabular}{c c c l}
        \hline
        \hline
        Postselection 
        & $C_\text{em} = \Lambda F^{-1}$ 
        & $C_\text{em}^\text{QPE}(\gamma, \epsilon) \leq \Lambda \kappa e/\pi \cdot {\gamma}/{\epsilon}$ 
        & 
        \parbox[m]{0.25\textwidth}{\raggedright \vspace{2pt} Symmetry verification,\newline verified phase estimation* \strut}
        \\
        \hline
        Rescaling
        & $C_\text{em} = \Lambda F^{-2}$ 
        & $C_\text{em}^\text{QPE}(\gamma, \epsilon) \leq \Lambda \kappa e/\pi \cdot 2 \gamma/\epsilon$ 
        & \parbox[p]{0.25\textwidth}{\raggedright \vspace{2pt} Echo verification*,\newline direct rescaling under GDN \strut}
        \\
        \hline
        Explicit unbiasing 
        & $C_\text{em} = \Lambda F^{-4}$ 
        & $C_\text{em}^\text{QPE}(\gamma, \epsilon) \leq \Lambda \kappa  e/\pi \cdot 4 \gamma/\epsilon$ 
        & \parbox[p]{0.25\textwidth}{\raggedright \vspace{2pt} Probabilistic error cancellation,\newline virtual distillation \strut}\\
        \hline
        \hline
    \end{tabular}
    \caption{
        Error mitigation overhead for expectation value estimation ($C_\text{em}$) and for quantum phase estimation ($C_\text{em}^\text{QPE}$) for different error mitigation techniques.
        $\Lambda$ and $\kappa$ are constants depending on the circuit, noise model and error mitigation approach, $\gamma$ is the noise rate and $\epsilon$ is the target precision of phase estimation.
        * The overhead of echo verification (a.k.a. verified phase estimation) depends on the assumptions on the circuit: if the cost of state preparation is negligible, its overhead is postelection-bounded; if the dominant cost comes from state preparation, it behaves like a rescaling technique due to the required doubling of the circuit depth.
        See App.~\ref{app:error_mitigation_overhead} for a detailed discussion.
    }
    \label{tab:em_overhead}
\end{table}

\textbf{Error mitigation beyond expectation values.} ---
Our third contribution is to develop a general EM protocol for
maximum-likelihood estimation.
This is motivated by our intention to apply MSQPE in the presence of arbitrary noise. 
However, our protocol is more general; it can be applied to any estimation problem that would require samples from a noiseless quantum circuit, while only samples from a noisy device with a known error channel are available.
Maximum-likelihood is a more general estimator than sample average, typically used for expectation value estimation.
To the best of our knowledge, prior EM schemes have exclusively targeted expectation value estimation, leaving the extension of EM to other estimation tasks as an open problem identified in the literature \cite{caiQuantum2023}.

Our protocol, detailed in App.~\ref{sec:EUMLE}, is based on rewriting the distribution $P(x)$ that describes the outcomes $x$ of a noiseless circuit as a non-convex combination of distributions $Q_j(x)$ that can be sampled on the noisy device: $P(x) = \sum_j \alpha_j Q_j(x)$.
This can be done by using the channel decomposition techniques from probabilistic error cancellation ~\cite{temmeError2017, endoPractical2018}.
Given a parameterized model $P(x|\phi)$ for the noiseless distribution, the expected value of the log-likelhood function $\bar{\mathcal{L}}(\phi)$ can then be expanded in terms of $Q_j$:
\begin{equation}
    \bar{\mathcal{L}}(\phi) = \sum_j \alpha_j \E[\log P(x|\phi) | x \sim Q_j(x)].
\end{equation}
This function can be estimated by taking samples from $Q_j(x)$ on the noisy quantum device, by importance sampling of $j$.
We obtain the Explicitly Unbiased Maximum Likelihood Estimator by maximizing this likelihood function, and we prove its convergence in the following theorem:

\begin{theorem}[Thm.~\ref{thm:eumle}, informal]
    Given a parameterized distribution model $P(x|\phi)$ that describes the outcomes $x$ of a noiseless quantum circuit, and assuming access to samples from a noisy quantum computer with a known error channel, we can construct an Explicitly Unbiased Maximum Likelihood Estimator (EUMLE).
    This estimator converges to the true value of $\phi$ with a variance that scales inversely with the total number of samples.
\end{theorem}

\textbf{QFT-based QPE with arbitrary noise.} ---
We then apply the EUMLE to MSQPE, obtaining a QFT-based QPE algorithm robust to arbitrary noise channels.
To ensure the estimator has a bounded sampling cost, we introduce a further regularization which addresses the issue of potential zeros of the likelihood function (see App.~\ref{sec:rEUMLE}).
Under local depolarizing noise at rate $\gamma$, the unitary-call cost of performing phase estimation to precision $\epsilon$ is found to be $\mathcal{T}_{\text{tot}}\leq  137 \gamma/\epsilon^2$.
This corresponds to a QPE error mitigation overhead of approximately $44\gamma/\epsilon$.
Comparing to the cost of MSQPE with global depolarizing noise in Theorem~\ref{thm:mle_sinqe_informal}, we observe that the additional overhead of mitigating a more general noise channel is approximately $7\times$ (see App.~\ref{sec:EUMLE_overhead}).

\textbf{Cost estimates.} ---
Finally, we turn to compiling MSQPE for a set of example problems to the surface code, for a range of residual logical noise rates $\gamma$.
Our goal is to quantify the trade-off between number of physical qubits and total computation time required to estimate ground state energies to a given precision.
The details of this resource estimates are given in App.~\ref{sec:compilation}, and we provide open-source code to reproduce and extend our results \cite{eftqpe_code}.

To perform this resource analysis, we need to select a specific QEC architecture, but we expect that our framework can be used to extend any such architecture into the early-fault tolerant regime.
We consider a quantum computer using a two-dimensional rotated surface code~\cite{fowlerSurface2012}, that performs fault-tolerant Clifford gates using lattice surgery~\cite{horsmanSurface2012} and uses CCZ magic state injection and distillation to implement Toffoli gates~\cite{gidneyEfficient2019}.
Such an architecture is popular because it can tolerate relatively high physical error rates, and because it can be implemented using a two-dimensional lattice of locally connected qubits~\cite{raussendorfFaulttolerant2007,fowlerHighthreshold2009}.

The target Hamiltonians we consider are (1) a set of $L\times L$-sites Fermi-Hubbard models and (2) electronic structure Hamiltonians for a selection of of small- to classically nontrivial molecules.
In both cases, we estimate the ground state energy to precision $\Delta E$ by applying MSQPE to a qubitized walk operator $\mathcal{W}$ \cite{lowHamiltonian2019}. 
For the Fermi-Hubbard Hamiltonian, we choose model parameters ($t=1, u=4$) and target precision ($\Delta E = 10^{-2}$) following Ref.~\cite{babbushEncoding2018}, reproducing a regime that is challenging for classical methods \cite{leblancSolutions2015}.
$\mathcal{W}$ is implemented following the methods of Ref.~\cite{babbushEncoding2018}.
Note that qubitization of the Hubbard model does not require data uploading by quantum read-only memories (QROMs), and the only approximation error comes from the rotations used to prepare a two-qubit state encoding the Hamiltonian coefficients.
The molecular models we consider span active space sizes between 6 and 26 spatial orbitals.
All molecular Hamiltonians were compressed with tensor hypercontraction (THC) \cite{leeEven2021} after applying symmetry shifts \cite{Rocca_2024} to further reduce the qubitization 1-norm $\lambda$.
We aim for the chemically desirable $\Delta E = {10^{-3}} E_\text{h}$ accuracy of the final results, choosing the THC hyperparameters accordingly. 
See App.~\ref{app:molecues_with:thc_details} for computational details on the construction of the molecular Hamiltonians.

We estimate the Toffoli costs of phase estimation on the qubitization walk operators $\mathcal{W}$ using Qualtran \cite{harriganExpressing2024}.
We consider using one or multiple CCZ factories \cite{gidneyEfficient2019} to produce magic states for implementing Toffoli gates,
and
We optimize the data-qubit code distances $d$ and the two-stage distillation code distances $d_0, d_1$ to minimize the computational volume of $c\mathcal{W}$ while keeping the error rate bounded by $\gamma$.
We then use the methods detailed earlier in the paper to optimize a robust estimator from oracle calls to the noisy $\mathcal{W}$, with a final target standard deviation of $\Delta\phi = \Delta E / \lambda$ (where $\lambda$ is the qubitization 1-norm).
(Note that this is different than requiring precision $\Delta\phi$ with a fixed success probability, which is a common choice in literature. We expand on this comparison in App.~\ref{sec:ftqc_comparison})

In Fig.~\ref{fig:ftqc-res-comparison}, we plot the resulting estimates of wall-clock time and number of physical qubits required to execute our algorithm. 
Across the range of experiments considered, we observe that circuit splitting allows for a factor $\sim 2$ decrease in the qubit count at the cost of a factor $\sim 10$ increase in wall-clock time.
All the curves demonstrate an ``elbow-like'' behaviour; reducing the number of physical qubits below a specific point incurs a significant penalty in total runtime.
The effect becomes more pronounced for larger systems.
This aligns with our back-of-the-envelope characterization of circuit division, where we observed there only is a thin region of error rate and number of available qubits where significant circuit depth enhancements are possible by combining error correction with error mitigation.
As anticipated, it is apparent that the simulation of chemistry is more resource-intensive than the simulation of the Hubbard model at the same state-space size.
Our analysis makes this quantitative: compared to the $L=5$ Hubbard model with $25$ spatial modes, the $26$ spatial orbital Co(salophen) molecule requires roughly four times more qubits and around one and a half orders of magnitude longer wall-clock time.
This is due to the higher complexity of molecular Hamiltonians and the costs associated to the QROMs requried to upload the Hamiltonian coefficients for realizing the qubitization oracle.

\begin{figure}
    \centering
    \includegraphics[width=\textwidth]{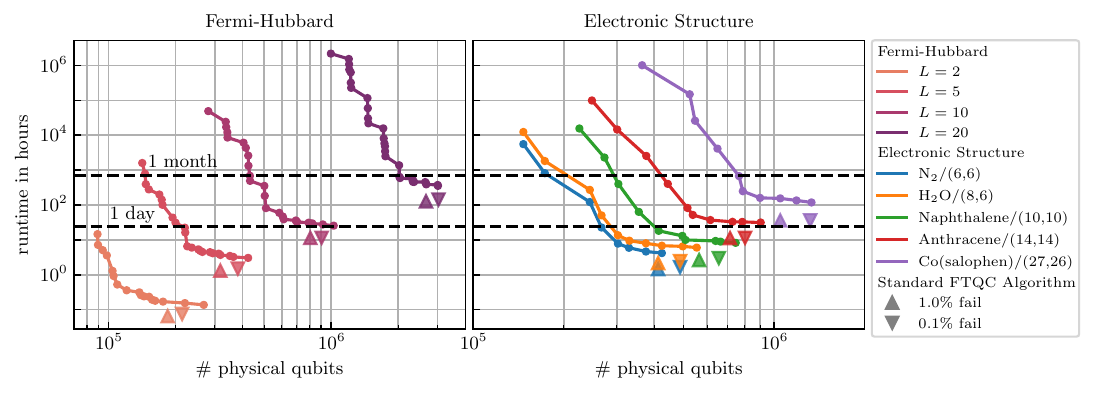}
    \caption{
        Physical costs for MSQPE applied to (left) the qubitized Fermi-Hubbard model Hamiltonian and (right) the active-space molecular Hamiltonians of some chemical systems.
        For the Fermi-Hubbard Hamiltonian, we choose the hopping parameter $t=1$ to set the energy units and interaction strength $u/t = 4$.
        For the electronic structure Hamiltonians, the active spaces sizes (number of electrons and spatial orbitals) are notated in the legend.
        The target precision on the energy is $\Delta E = 10^{-2}$ for the Hubbard model and $\Delta_E = 10^{-3} E_\text{h}$ for the chemical systems.
        The physical resources are estimated assuming computation in the surface code using CCZ resource states produced by a single factory; the size of surface code and CCZ factory parameters are chosen to minimize computation volume while keeping the error rate below a choosen $\gamma$, assuming physical error rate of $10^{-3}$, a surface code clock cycle of $\SI{1}{\micro\second}$, and $50\%$ routing overhead.
        The residual error is assumed to follow a global depolarizing noise model.
        Changing $\gamma$ allows to trade between the number of physical qubits (used as resource for implementing better error correction) and total runtime.
        For comparison, we report the physical costs of the standard fault-tolerant single-circuit implementation of sin-state QPE, accepting a failure probability of $1\%$ or $0.1\%$ (triangles).
    }
    \label{fig:ftqc-res-comparison}
\end{figure}

\section{Discussion and Conclusion}
Besides various technical contributions to phase estimation and error mitigation, the key focus of this work is to develop a framework that adapts fully fault-tolerant algorithms to the early-fault-tolerant regime and quantifies the savings achieved by trading error rates for physical qubit counts.
Such a framework requires a full-stack treatment of quantum algorithm compilation, including a robust treatment of the output of a noisy quantum computer.
Improvements in logical gates, QEC codes, or Hamiltonian simulation algorithms can be immediately tested within our code, which relies on Qualtran~\cite{harriganExpressing2024} for computing logical costs, in a plug-and-play fashion.
The importance of this adaptability is highlighted by recent advances in T-gate cultivation~\cite{gidneyMagic2024}, which will significantly reduce the resource estimates presented here and necessitate going beyond simple T-gate counting.
Additionally, our approach allows for exploring alternative QEC architectures, such as those with long-range connectivity and using non-surface codes~\cite{bluvsteinLogical2024,silvaDemonstration2024}.
As illustrated in Fig.~\ref{fig:ftqc-res-comparison}, with current methods, qubit counts can be reduced by a factor $2-3$, but with significant overheads of a factor $10$-$100$ in runtime.
This result, while somewhat sobering for early fault-tolerance research, stems from the fact that the gain in physical qubit count from increased logical circuit depth [Eq.~\eqref{eq:qubit_savings}] is only logarithmic.

We would like this paper to be understood as an invitation to try and beat the provided early-FT estimates, and we can already suggest various improvements that could be achieved within our framework.
In this work, we have focused on developing noise-robust post-processing for data sampled using the sine-state QPE circuit, which is optimal only in the noiseless case.
Optimizing this circuit while taking noise into account could further improve MSQPE and bring its performance closer to the information-theoretic limit \cite{kolodynski2013efficient}, potentially saving a factor $\approx 2$ (see App.~\ref{sec:qpe_gdn_numerical_comparison}).
Additionally, applying the unary-iteration-based technique from Ref.\cite{leeEven2021} to the iterated qubitized walk operator could yield another $2\times$ reduction in logical costs.
Furthermore, in App.~\ref{sec:filtering} we outline a filtering procedure for the EUMLE, which could reduce the error mitigation overhead towards the bound from postselection, potentially achieving a $4\times$ reduction in total cost (analogous to Thm.~\ref{thm:em_overhead_qpe_informal}).
Additional savings could also be achieved in the compilation of our algorithms, by allowing higher algorithmic error.
For instance, reducing the precision of rotation angles could lower physical resource requirements, though balancing this against energy estimation accuracy remains nontrivial.

We emphasize that the EUMLE introduced in this work has potential applications beyond QPE, in more general estimation problems.
On the other hand, the EUMLE needs accurate characterization of the noise model in order to construct a quasiprobability decomposition.
Investigating the robustness of EUMLE against inaccurate noise characterization is a clear direction for future work.
Similarly, further research is needed to characterize residual logical noise in specific fault-tolerant computation models.
This characterization could be carried out analytically, by analyzing error correcting codes, decoders and distillation protocols, or experimentally, by extending noise characterization techniques such as gate set tomography \cite{merkel2013self,blume2017demonstration,Nielsen2021gatesettomography} to fault-tolerant computation.

We have aimed to be as complete as possible in the resource estimate of MSQPE.
However, we have neglected a significant complicating factor for ground state energy estimation: the preparation of a good initial state.
Efficient protocols for preparing high-fidelity initial states are a topic of active research in the field of quantum simulation for chemistry~\cite{tubmanPostponing2018, leeEvaluating2023,fomichev2024initialstatepreparationquantum, berry2024rapidinitialstatepreparation},
and the cost of preparing a state $|\psi\rangle$ with overlap $a=|\langle\phi|\psi\rangle|$ with the ground state $|\phi\rangle$ can vary widely depending on the problem.
If only an approximate ground state is available (i.e. $a\neq 1$) the energy can still be estimated, provided that the first excited state of the Hamiltonian is well separated from the ground state, by an energy gap $\Delta$.
This can be achieved in two ways:
(1) energy-filtering and amplitude amplification, which allow to disill the exact ground state $|\phi\rangle$ \cite{geFaster2019,linNearoptimal2020}; or (2)  classical signal processing techniques that isolate the ground-state phase information~\cite{obrien2019quantum,linHeisenbergLimited2022,dutkiewicz2022heisenberg,dingEven2023,dingSimultaneous2023}.
While the first method integrates immediately with our algorithm, it requires long state-preparation circuits, with a cost scaling as $\mathcal{O}(\Delta^{-1}a^{-1})$.
This cost needs to be added for each circuit sample, which penalizes circuit division.
The second class of methods require only one preparation of $|\psi\rangle$ for each circuit sample, but introduce a sampling overhead scaling as $\mathcal{O}(a^{-2})$ and require QPE circuits of at least depth $\mathcal{O}(\Delta^{-1})$~\cite{moitraSuperresolution2015}. 
Drawing inspiration from these methods, we propose a filtered EUMLE algorithm in App.~\ref{sec:filtering}, which we anticipate will achieve similar performance; determining its precise overhead is left to future work.
All the options described above establish a lower bound on the minimum circuit depth required for QPE of $\Omega(\Delta^{-1})$, which translates in our framework to a maximum tolerable error rate $\gamma$ (and thus a minimum required number of qubits).
If the error rate exceeds this threshold, the cost of estimating the ground state energy becomes exponential in $\epsilon$~\cite{moitraSuperresolution2015}.
Therefore, we suggest that early fault-tolerant algorithm research should prioritize optimizing initial state preparation and identifying applications where state preparation costs is reasonably low.

\section*{Code availability}
Python code for performing the logical and physical resource estimates of MSQPE as well as all used molecular Hamiltonians and their factorizations is publicly available at {\url{https://github.com/StefanoPolla/EarlyFT_QPE}}.

\section*{Contributions}
A.~D., S.~P. and T.~E.~O. conceived the main ideas, developed the algorithms and conceived and proved the theorems.
S.~P. and T.~E.~O. conceived the error mitigation scheme.
W.~J.~H. provided the performance estimate for early fault tolerance in App.~\ref{sec:early_ft_algs}.
S.~P., A.~D. and M.~S. developed the cost estimates code.
M.~S. provided the molecular models for the chemistry cost estimates.
All authors contributed to interpreting the results, writing, and revising the manuscript.

\begin{acknowledgments}
We acknowledge useful discussions with Oumarou Oumarou, Kianna Wan, Ryan Babbush and Rafał Demkowicz-Dobrzański, we thank Rolando Somma and Nobuyuki Yoshioka for useful feedback on the manuscript, and we thank Carlo Beenakker for support.
We thank the Lorentz Center for hosting the authors during the workshop ``Bridging the gap between classical \& quantum simulation'', enabling the progress of this work.
S.P.~acknowledges support from Shell Global Solutions BV during his time in Leiden. 
A.D.~is supported by a Google PhD Fellowship.
\end{acknowledgments}

\clearpage
\appendix

\renewcommand{\thetheorem}{\thesection.\arabic{theorem}}

\let\oldsection\section

\renewcommand{\section}[1]{%
  \oldsection{#1}
  \setcounter{theorem}{0}%
  \renewcommand{\thetheorem}{\thesection.\arabic{theorem}}%
}

\section{Designing algorithms for early fault tolerance}\label{sec:early_ft_algs}

Most analysis of fault-tolerant quantum algorithms makes the assumption that quantum error correction (QEC) enables error-free computation.
Under reasonable physical assumptions, the space and time overheads of QEC are at most polylogarithmic in the size of the computation.
In the limit where devices are large and error rates are low, this asymptotic behavior dominates and it is affordable to suppress errors to arbitrarily small levels.
However, before we reach this regime, we expect an era of ``early fault-tolerance,'' where we have access to quantum computers that are capable of benefiting from quantum error correction, but are not so large and performant that we can abstract away noise entirely.

Some works on early fault-tolerance consider other limitations besides size and error rate that might impose limits on quantum error-correction~\cite{katabarwaEarly2023}.
For example, quantum error correction is designed to suppress uncorrelated errors, but early experimental evidence with superconducting qubits has demonstrated that large correlated errors can arise from cosmic ray impact events~\cite{mcewenResolving2022, acharyaSuppressing2023}.
Experimental progress has shown that this effect can be suppressed~\cite{panEngineering2022, mcewenResisting2024}, although unknown mechanisms still appear to set limitations on the performance of quantum error correction in state-of-the-art experiments~\cite{acharyaQuantum2024}.
In this work, we generally make the assumption that scalable quantum error correction is possible and that any confounding factors will be addressed by scientific and engineering progress.
However, it is possible that the early fault-tolerant era will be prolonged by unforeseen difficulties, in which case the techniques we explore may prove especially useful.

There is a growing body of work focused on understanding and maximizing the power of early fault-tolerant quantum computers.
One frequently-taken approach is to assume that some combination of techniques will allow us to execute much larger circuits than we can on today's noisy devices, but that practical concerns will still impose limits on the number of qubits and the circuit size.
Works in this direction focus on applications, usually picking a simple computational task and exploring different design choices in order to minimize the resources required to perform the task in the absence of error.
The focus on early fault-tolerance in these works appears mainly in the choice of the computational task, which might be a toy example designed to demonstrate algorithmic primitives~\cite{bluntCompilation2024}, or a scientific application just beyond the reach of classical computation~\cite{childsFirst2018,campbellEarly2022}.
When the constant factors are taken into account, and when applicable (finite) target precisions are fixed, it is common to find that the simpler methods with worse asymptotic scaling (such as Trotterized time evolution) are superior to their asymptotically optimal counterparts (such as quantum signal processing)~\cite{childsFirst2018,campbellEarly2022}.

Other works have focused on developing new algorithmic tools for the setting where we have large but finite upper bounds on the maximum circuit size~\cite{wangQuantum2023, wangFaster2023, linHeisenbergLimited2022, zhangComputing2022, wanRandomized2022, dongGround2022, dingEven2023, dingSimultaneous2023, niLowdepth2023, nelsonAssessment2024, akhalwayaTopological2024}.
These algorithms focus on reducing the number of gates, the circuit depth, or the number of ancillary qubits by classically combining the results from multiple separate circuit executions.
Underlying many of these approaches is a trade-off between using deeper circuits to estimate some quantity at the Heisenberg limit and using more repetitions of a shallow circuit to estimate the same quantity at the shot-noise limit.
For example, Refs.~\citenum{wangQuantum2023, wangFaster2023, dingEven2023} all allow for ground state energy using a total runtime that scales with the precision \(\epsilon\) as \(\epsilon^{-\alpha}\) for an \(\alpha \in \left[ 1,2 \right]\) that decreases as the maximum circuit depth is increased.
Other works make a related trade-off, choosing between running many copies of a shorter circuit and postselecting on observing some rare event, or obtaining a quadratic speedup with amplitude amplification and the deeper circuits it requires.~\cite{dongGround2022,akhalwayaTopological2024}.

\subsection{Combining quantum error mitigation and quantum error correction}
\label{sec:combining_EM_and_qec}

Early fault-tolerant quantum computers won't be able to arbitrarily suppress error rates with a negligible resource overhead, which makes it desirable to consider computations that can tolerate a non-zero error rate.
This can be achieved either by developing algorithms that have some natural robustness to error~\cite{hugginsUnbiasing2022,kshirsagarProving2024,liangModeling2024}, or by augmenting error correction with quantum error mitigation (EM)~\cite{suzukiQuantum2022, piveteauError2021, lostaglioError2021, akahoshiPartially2024}.
Years of work in the NISQ era has produced a large body of work on algorithms robust to error and we do not attempt to review that literature here.
For similar reasons, we limit our discussion of error mitigation to those works that explicitly explored the combination of quantum error correction and error mitigation.
Broadly, error mitigation allows us to estimate noise-free expectation values at the expense of increasing the overall computation time exponentially in the number of expected errors~\cite{caiQuantum2023}.
For example, given a known noise model, probabilistic error cancellation lets us express the noise-free expectation value of an observable as a quasi-probability distribution over quantities sampled from modified versions of the original circuit~\cite{temmeError2017}.
In this case, the exponential overhead manifests as an increased variance, or, equivalently, as a larger number of samples required to estimate an expectation value to a fixed precision.

The combination of quantum error mitigation and quantum error correction can take a number of different forms.
Motivated by the fact that fault-tolerant non-Clifford gates are particularly costly to implement, some works have focused on combining error-corrected Clifford gates with noisy non-Clifford gates~\cite{piveteauError2021,suzukiQuantum2022, toshioPractical2024}.
By using error mitigation to handle the noise in the non-Clifford gates, such schemes avoid the difficulties and overheads of, e.g., magic state distillation.
Other works have explored the use of error mitigation for both Clifford and non-Clifford gates, potentially allowing the use of lower code distances and fewer qubits.
In some cases, researchers have explicitly considered the optimization of quantum error correction parameters for use in conjunction with error mitigation~\cite{katabarwaEarly2023,suzukiQuantum2022,akahoshiPartially2024}.

To develop some intuition, we illustrate a toy example of how QEC and EM can be combined in Fig.~\ref{fig:NISQ_crossover}.
In the top panel of this figure, we estimate the depth of the largest circuit that one could perform on \(100\) logical qubits for a variety of physical error rates and numbers of (physical) qubits.
In this model, we demand an ``effective sampling rate'' of one sample per minute. 
In other words, the combination of error correction and error mitigation should emulate the effect of an idealized error-free quantum computer that executes the circuit and performs a measurement of a desired observable once per minute.
We approximate the effective time per sample using the expression \(N_{\mathrm{parallel}}^{-1} D t_{\mathrm{gate}} \gamma^{2}\), where \(N_{\mathrm{parallel}}\) is the number of copies of the computation that we can perform in parallel given our physical resources, \(D\) is the circuit depth, \(t_{\mathrm{gate}}\) is the time per gate, and \(\gamma^2\) is the overhead required to mitigate the logical errors using probabilistic error cancellation.
We optimize over the choice of whether or not to use error correction and the code distance \(d\), finding the choice that allows for the largest \(D\) such that the effective time per sample is less than one minute.

\begin{figure}
    \centering
    \includegraphics[width = 0.9\textwidth]{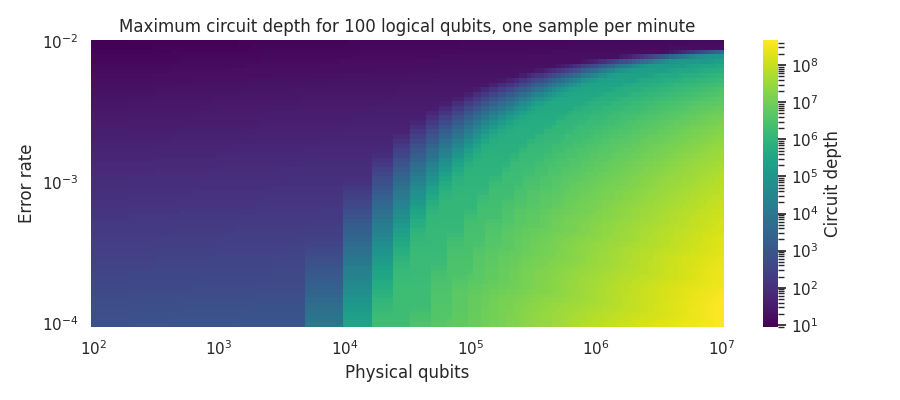}
    \includegraphics[width = 0.9\textwidth]{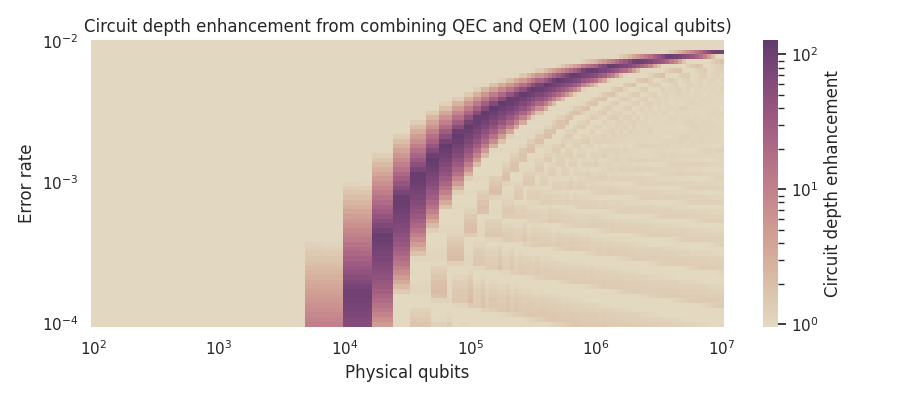}
    \caption{(top) Maximum circuit depth for a circuit with 100 logical qubits and an effective sampling rate of one sample per minute over a variety of physical error rates and qubit numbers. 
    We assume a circuit composed entirely of CNOT gates, a raw two-qubit gate time of 20 nanoseconds, and a surface code cycle time of 970 nanoseconds. 
    Errors are mitigated using probabilistic error cancellation (PEC), or by a combination of surface code quantum error correction and PEC.
    We maximize the achievable depth by optimizing over the choice of whether or not to use error correction, and the code distance.
    (bottom) Increase in circuit depth allowed by this combination of QEC and EM, when compared with a protocol that either uses EM or QEC (with QEC parameters chosen such that the error probability throughout the circuit is less than \(1\%\)). 
    The existence of the dark band reveals a regime where this increase is substantial.
    }
    \label{fig:NISQ_crossover}
\end{figure}

For simplicity, in this example we assume a circuit composed entirely of CNOT gates executed on a superconducting quantum processor similar to the one described in Ref.~\cite{acharyaSuppressing2023}.
We model the error suppression factor of the surface code using the expression
\begin{equation}
    \Lambda = \frac{0.01}{p_{phys}},
    \label{eq:Lambda_def}
\end{equation}
where \(0.01\) is an estimate of the surface code threshold and \(p_{phys}\) is a single parameter that describes the error rate of physical gate and measurement operations~\cite{fowlerLow2019}.
This means that we can estimate the error per surface code cycle as \(0.1 \Lambda^{-(d + 1)/ 2}\), where the constant \(0.1\) is an estimate of the true prefactor.
We can perform CNOT gates using lattice surgery using two qubits operating for \(3\,d\) surface code cycles, so we can therefore estimate the logical error probability for a gate as
\begin{equation}
    p_{logical} =
    \begin{cases}
      p_{phys} & \text{ if \(d=1\)} \\
      1 - \left(1 - 0.1 \frac{0.01}{p_{phys}}^{-(d+1)/2}\right)^{6d} & \text{ if \(d > 1\)},
    \end{cases}
\end{equation}
using the convention that \(d=1\) indicates that we are not using quantum error correction at all.
Using a physical gate time of \(20\) nanoseconds and a surface code cycle time of \(970\) nanoseconds~\cite{acharyaSuppressing2023}, we therefore have
\begin{equation}
    t_{\mathrm{gate}} = 
    \begin{cases}
        20 \cdot 10^{-9} & \text{ if \(d=1\)} \\
        3 \cdot d \cdot 970 \cdot 10^{-9} & \text{ if \(d > 1\)}.
    \end{cases}
\end{equation}
A logical qubit in the surface code requires \(2d^2 - 1\) physical qubits, so we can execute 
\begin{equation}
    N_{\mathrm{parallel}} = \left\lfloor{\frac{N_{phys}}{100 \left( 2 d^2 - 1 \right)}}\right\rfloor
\end{equation}
copies of a \(100\)-logical qubit computation in parallel.
Treating \(p_{logical}\) as the probability of applying a two-qubit depolarizing channel, we can follow Ref.~\cite{temmeError2017}, finding that the overhead incurred by applying probabilistic error cancellation is given by 
\begin{equation}
    \gamma^2 = \frac{1 + 7/8 p_{logical}}{1 - p_{logical}}^{2G},
\end{equation}
where \(G = 49.5 D\) is the number of noisy gates in the circuit.

In this simplified model, we can attempt to quantify the benefit of combining quantum error mitigation and quantum error correction.
To do so, we perform a second calculation to determine the maximum depth achievable when either (i) using physical gates and error mitigation alone or (ii) using quantum error correction alone.
For (i), we calculate the maximum depth as above, restricting \(d=1\).
For (ii), we calculate \(p_{logical}\) as above and choose the largest \(D\) such that the overall probability of an error in the circuit is less than \(1\%\), i.e.,
\begin{equation}
    D_{\text{QEC max}} = \max \left\{ D : 1 - \left( 1 - p_{logical} \right)^{49.5 D} \leq .01 \right\}.
\end{equation}
In the bottom panel of Fig.~\ref{fig:NISQ_crossover}, we plot the ratio of the circuit depth achievable with the combination of EM and QEC compared to the circuit depth achievable using either technique alone.
We see that there is a regime on the boundary between NISQ and full fault-tolerance where the combination of QEC and EM allows for a substantial enhancement in the circuit depth.

\subsection{Circuit division in early fault-tolerance}

In the previous section we explored the ability of quantum error mitigation to boost the maximum depth allowed in a (logical) quantum circuit when being executed on a fixed number of physical qubits and a fixed error rate.
However, this still bounds the circuit depths available to quantum algorithmists on a fixed quantum device to some $D_{\max}$.
In principle one can consider running multiple quantum circuits $j$ of depth $D_j<D_{\max}$, given access to a classical post-processing method that is robust against the failure of some of these circuits.
Perhaps surprisingly, such circuits exist for many problems.
The main contribution of our paper is to explore this ``circuit division'' in the context of on such problem class --- quantum phase estimation (QPE) --- to combine this with optimized error mitigation techniques, and compile the resulting algorithms to a realistic physical implementation.

Before introducing the relevant background on QPE and explaining our results in detail, let us roughly estimate the potential benefit of circuit division in an early-fault-tolerant setting.
\begin{definition}
    A circuit division of a quantum algorithm $\mathcal{A}$ that requires a circuit of $G$ gates corresponds to a set of $M$ ``sub-circuits'' of length $\leq G/R$, and a classical post-processing routine which can be applied to the sub-circuits to solve the problem targeted by $\mathcal{A}$.
\end{definition}
We target phase estimation routines, for which (under some problem constraints~\cite{moitraSuperresolution2015}), this division is allowed at a cost $M\sim R^2$.
For a circuit with \(G\) gates, we can approximate the overall probability for the circuit to execute without error as \(1 - p_{circuit} = 1 - (1 - p_{logical})^G \approx e^{-G p_{logical}}\), where \(p_{logical}\) is the error rate per gate.
Each logical qubit in the surface code requires \(2d^2 + 1\) physical qubits (not accounting for overheads from routing or magic state distillation), where the code distance \(d\) is an odd number that we can choose freely.
For a given physical error rate \(p_{phys}\), we can approximate the error suppression factor of the surface code \(\Lambda\) as in Eq.~\ref{eq:Lambda_def}, which implies an error per logical qubit per cycle proportional to \(\Lambda^{-((d+1)/2)}\).
Without specifying a particular choice of gate and compilation into surface code operations, we can still make the approximation that \(p_{logical} \propto d \Lambda^{-\left(d + 1\right)/2}\) (since most gates require a number of cycles proportional to \(d\) and the error rates involved are small).
The constant of proportionality is determined by the details of the error model and the particular choice of gates, but we can set it to \(1\) to obtain a qualitative understanding, yielding
\begin{equation}
    1 - p_{circuit} \approx \exp \left(-G d \Lambda^{-\left(d + 1\right)/2}\right).
    \label{eq:p_success_sketch}
\end{equation}
This exponential decay puts a cutoff on the circuit sizes that we can accept before the failure probability grows too large.

A protocol that replaces a single execution of a circuit with \(G\) gates by \(R^2\) executions of circuits with \(G/R\) gates increases the overall time required to obtain a solution by a factor of \(R\).
If we demand a constant error rate per circuit and neglect logarithmic terms, then Eq.~\ref{eq:p_success_sketch} suggests that we should take \(G \propto \Lambda^{(d+1)/2}\).
To obtain a concrete estimate, let us consider \(\Lambda = 10\) and \(G = 10^6\), which implies that we need a code distance of \(d=11\).
If we let \(d'\) be the code distance required for a circuit with \(\frac{G}{R}\) gates, then we have \(d - d' \approx 2 \log_{\Lambda} R\).
The number of physical qubits required per logical qubit is approximately proportional to \(d^2\), so taking \(N'_{phys} \propto (d')^2\) and \(N_{phys} \propto d^2\), we can approximate the savings in terms of the number of physical qubits as
\begin{equation}\label{eq:qubit_savings}
    \frac{N'_{phys}}{N_{phys}} \approx 1 - \frac{4 \log_{\Lambda} R}{d} + \frac{4 \log_{\Lambda}^2 R}{d^2}.
\end{equation}
If these approximations hold, then we could achieve a factor of two reduction in the number of physical qubits with \(R \approx 40\), but it would appear that a factor of ten reduction would require \(R\) to be unfeasibly large.
While the model of errors we considered above is oversimplified, we shall show that a more detailed analysis yields a qualitatively similar conclusion.

\section{Background on quantum phase estimation}\label{sec:qpe_background}

The term `quantum phase estimation' (QPE) refers to a family of computational problems related to estimating eigenphases $\phi_j$ of a unitary operator $U$, $U\ket{\phi_j}=e^{i\phi_j}\ket{\phi_j}$, to some target precision $\epsilon$.
In this work, we target the estimation of a single, specific eigenvalue $\phi$; e.g. if $U=e^{iHt}$ for a Hamiltonian $H$, one might target the lowest eigenvalue of $H$.
(Note that much recent work on QPE considers the estimation of multiple eigenvalues, see e.g. Refs.~\cite{obrien2019quantum,sommaQuantum2019,linNearoptimal2020,dutkiewicz2022heisenberg,linHeisenbergLimited2022}.)
Following common practices in metrology~\cite{babbushEncoding2018}, we define `precision' to mean a bound on the Holevo error of the error in our estimator $\widetilde{\phi}$:
\begin{equation}
    \epsilon_H = \sqrt{\E[4\sin^2(\frac{\widetilde{\phi}-\phi}{2})]}
\end{equation}
(as opposed to the computer science literature~\cite{nielsen2001quantum} where one typically requires $P(|\phi-\widetilde{\phi}|\leq\epsilon)\geq (1-p)$ for fixed $p$).
This metrological definition is slightly stricter~\cite{kimmelRobust2015} as it requires bounding the tail of the distribution of the estimator $\widetilde{\phi}_j$.

The problem of determining ground states is known to be hard even with a quantum computer (i.e. QMA-hard) for even $2$-local Hamiltonians \cite{kempeComplexity2006}), but it reduces to a problem solvable by a quantum computer (i.e. one in BQP) when one assumes access to an initial state $|\psi\rangle$ which has large overlap $\langle\psi|\phi\rangle$ with the target eigenstate~\cite{wocjanSeveral2006,gharibianDequantizing2023}.
In this work, we further require that our target eigenvalue $\phi$ is separated from other eigenvalues $\phi_j$ (at least those for which $|\langle\phi_j|\psi\rangle| > 0$) by a gap $\min_j|\phi-\phi_j|=:\Delta \gg\epsilon$.
This is necessary by information-theoretic constraints~\cite{moitraSuperresolution2015} which set $\Delta^{-1}$ as a lower bound on the required evolution by controlled-$U$ in order to resolve an eigenvalue (without an overhead that grows exponentially in the number of nearby eigenvalues).
This in turn places an upper bound on the amount by which we can divide a circuit, which will later coincide with a minimum fidelity requirement on a noisy implementation of $U$.
Under this assumption, one can project the starting state $|\psi\rangle$ onto the eigenstate $|\phi\rangle$ at a cost proportional to $\Delta^{-1}|\langle\psi|\phi\rangle|^{-1}$~\cite{geFaster2019,linNearoptimal2020}.
In this work we absorb this cost onto the cost of initial state preparation.

\begin{definition}[QPE for single eigenstates]
\label{def:problem_QPE}
    Let $U$ be a unitary operator with eigendecomposition $U|\phi_j\rangle=e^{i\phi_j}|\phi_j\rangle$.
    Assume access to a preparation unitary $V_{\ket{\phi}}$ for an eigenstate $|\phi\rangle$ of $U$, and access to a (possibly noisy) implementation $\widetilde{\mathcal{U}}$ of controlled-$U$ ($U_c$), and fix some $\epsilon > 0$.
    The QPE problem is to produce an estimator $\widetilde{\phi}$ of $\phi$ with Holevo error $\epsilon_H\leq \epsilon$.
    We say that the (oracular) cost $\mathcal{T}_\text{tot}$ of solving the QPE problem is the total number of applications of $\widetilde{\mathcal{U}}$ required to implement the estimator.
    $\mathcal{T}$ is the maximum number of applications of $\widetilde{\mathcal{U}}$ in a single quantum circuit.
\end{definition}

In this section, we review the well-established algorithms to solve the noiseless case of this problem (in preparation for adding noise in the remainder of the text):
\begin{definition}[noiseless QPE for single eigenstates]
    \label{def:noiseless_QPE}
    The problem in Definition~\ref{def:problem_QPE} in the absence of noise in the quantum circuit, i.e.~$\widetilde{\mathcal{U}} = U_c^\dagger (.) U_c$.
\end{definition}

Existing quantum algorithms for phase estimation can be divided into two classes: \emph{Quantum Fourier Transform (QFT)-based} methods (also called \emph{parallel} or \emph{entanglement-based}) and \emph{single-control based} methods (also called \emph{sequential}, \emph{iterative} or \emph{single-control}).
In the first class, the phase estimate is obtained from a single run of a circuit with multiple control qubits.
Improving the precision requires increasing the dimension of the control register and the circuit depth.
In the second class, the phase is extracted by classically processing expectation values of multiple simpler circuits.
Better precision can be achieved either by increasing the depth of the circuits, or the number of samples used to estimate the expectation values.
We describe both classes in more details in the following 2 sections.

\subsection{QFT-based QPE}

\topic{QFT-based QPE} 
In this approach we obtain the first $n$ bits of the binary expansion of the phase from a single measurement of an additional $n$-qubit register.
The circuit consists of three stages: first, a probe state is prepared; then, the register is used to control the $U$; finally QFT.
Using $n$ control qubits yields precision of $\epsilon = \mathcal{O}(2^{-n})$ \cite{nielsen2001quantum}.

Possibly the most well-known variant of phase estimation is that presented in Nielsen and Chuang~\cite{nielsen2001quantum}, based on previous work from~\cite{cleve1998quantum}; this is sometimes known in the field as `textbook' QPE.
This algorithm uses two registers, a $\log_2(K)$-qubit control register and a system register.
The first step of the algorithm is preparing a ``probe state'' on the control register (in the textbook algorithm this is done by the Hadamard transform, preparing $\frac{1}{\sqrt{k}}\sum_k^K \ket{k}$).
The second step is applying a repeated time-evolution oracle $U^k$ on the system register, with $k$ being controlled by the control register state. 
Finally, an inverse quantum Fourier transform is applied before measuring the control register in a computational basis.
The output of this measurement is a binary expansion of the phase estimate.

\topic{Optimal probe state}
It has been shown that, in the absence of noise, the estimator constructed following this algorithm is optimal in terms of accuracy when there is no prior information on the phase to estimate \cite{vandamOptimal2007}.
The probe state that optimizes the Holevo variance of the estimator, however, is not the uniform superposition $\frac{1}{\sqrt{K}}\sum_{k=0}^{K-1} \ket{k}$ but rather the state
\begin{equation}
    \ket{s_K} = \sqrt{\frac{2}{K+1}}\sum_{j=0}^{K-1} \sin\left(\frac{j+1}{K +1}\pi\right) \ket{j},
\end{equation}
which can be prepared with cost $\mathcal{O}(\log{K})$ \cite{babbushEncoding2018} (see Fig.~\ref{fig:sinqpecircuit}).
The Holevo error of this estimator is $\tan(\frac{\pi}{K+1})\approx \pi/K$ \cite{berry2000optimal}.
This algorithm has been the most utilized in fault-tolerant application compilation research due to this optimality\cite{babbushEncoding2018,vonburgQuantum2021,leeEven2021}.

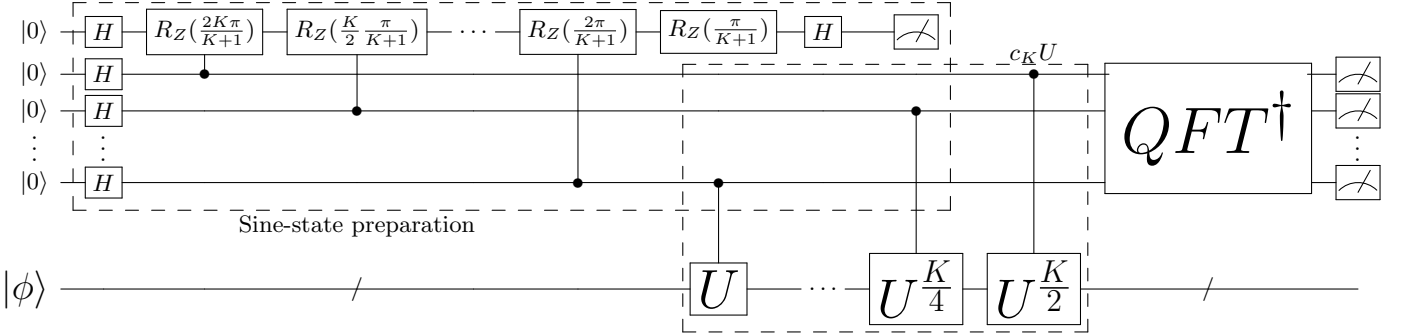
\begin{figure}[ht]
\caption{
    Sin-state QPE circuit with control dimension $K$. First, the state $\ket{s_K}$ is prepared on the control register of $\lceil\log_2K \rceil$ qubits
    (conditional on $1$ measured the ancillary qubit)
    , and the system register is prepared in state $\ket{\phi}$. Then, controlled unitary $c_K U = \sum_{k=0}^{K-1} \ket{k}\bra{k} \otimes U^k$ is applied. Finally, $QFT^\dagger$ is applied on the control register, and the register is measured in the computational basis to obtain a bit string $x$. The oracular cost is $K-1$.
}

    \begin{equation*}
    \Qcircuit @C=1em @R=.1em {
     \lstick{\ket{0}} &\gate{H}&\gate{R_Z(\frac{2K\pi}{K+1})}&\gate{R_Z(\frac{K}{2}\frac{\pi}{K+1})}&\push{\rule{0.15em}{0em}\dots\rule{0.15em}{0em}}\qw &\gate{R_Z(\frac{2\pi}{K+1})}&\gate{R_Z(\frac{\pi}{K+1})}& \gate{H}&\meter&\dstick{c_K U} \\
    \lstick{\ket{0}} & \gate{H}&\ctrl{-1}&\qw&\qw &\qw&\qw&\qw&\qw& \ctrl{5} & \multigate{3}{\text{\Huge $QFT^\dagger$}}& \meter\\ 
    \lstick{\ket{0}} &\gate{H}&\qw&\ctrl{-2}&\qw &\qw&\qw&\qw& \ctrl{4} &\qw& \ghost{\text{ \Huge $QFT^\dagger$}}& \meter\\
    \lstick{\rotatebox{90}{$\cdots$}\ }&\push{\rotatebox{90}{$\cdots$}}&&&&&&&&&&\rotatebox{90}{$\cdots$}\\
    \lstick{\ket{0}} & \gate{H}&\qw&\qw&\qw &\ctrl{-4}&\ctrl{2} &\qw&\qw&\qw& \ghost{\text{ \Huge $QFT^\dagger$}}& \meter\\ 
    \push{\rule{0em}{2em}}&&&\mbox{\text{Sin-state preparation}}\\
    \lstick{\text{\Large$\ket{\phi}$}}&\qw&\qw&{/}\qw&\qw&\qw& \gate{\text{\huge $U$}} & \push{\rule{0.15em}{0em}\dots\rule{0.15em}{0em}}\qw & \gate{\text{\huge $U^{\frac{K}{4}}$}} & \gate{\text{\huge $U^{\frac{K}{2}}$}} & \qw{/}&\qw
    \gategroup{1}{2}{5}{9}{1em}{--}
    \gategroup{2}{7}{7}{10}{0.6em}{--}
    }
    \end{equation*}

    This circuit produces a random variable $x \in {0, 1, ..., K-1}$ with probability distribution
    \begin{equation} \label{eq:sinqpe_noisless_prob_distribution}
    P^{(SinQPE)}(x | \phi) = \frac{\sin^2\frac{\pi}{K+1}}{K(K+1)} \frac{1+\cos[(K+1)(\phi - 2\pi \frac{x}{K})]}{(\cos (\phi - 2\pi \frac{x}{K}) - \cos\frac{\pi}{K+1})^2}.
    \end{equation}
\label{fig:sinqpecircuit}
\end{figure}

\begin{algorithm}
\label{alg:sine-state-single-circuit}
    [Single circuit sin-state phase estimation algorithm]
    Input: target precision $\epsilon_\mathrm{t}$, oracle access to initial state $\ket{\phi}$ and $cU$.
    \begin{enumerate}
        \item Let $n = \lceil\log_2 (\pi\arctan(\epsilon)^{-1}-2) \rceil$, $K = 2^n$.
        \item Run SinQPE circuit (Fig.~\ref{fig:sinqpecircuit}) with control dimension $K$ to obtain $x$. 
        \item Output $\widetilde\phi = 2\pi \frac{x}{K}$.
    \end{enumerate}
    [Note that it is also possible to choose $K = \left\lceil\arctan(\frac{\pi}{\epsilon_\mathrm{t}}-2)\right\rceil\neq 2^n$. This is discussed later in Fig.~\ref{fig:sinqpecircuit_arb_dim} where we explicitly construct the sin-state QPE circuit for general control dimension $K\in\NN$ with a log-log overhead.]
\end{algorithm}
We distinguish here between the ``target precision'' $\epsilon_\mathrm{t}$ (which is the number given to the algorithm that defines the input parameter $\Delta$), and the actual precision $\epsilon$ of the algorithm.
As we will see in Sec.~\ref{sec:qpe-global-depol}, these parameters are no longer the same in the presence of noise.

\begin{theorem} 
\label{thm:sinqpe-noiseless}
\cite{berry2000optimal}
    Alg.~\ref{alg:sine-state-single-circuit} with $\epsilon_\mathrm{t}=\epsilon$ solves Problem~\ref{def:noiseless_QPE} with cost $\mathcal{T}_{\mathrm{tot}} = \mathcal{T} = K-1 \approx \pi/\epsilon$, the lowest possible $\mathcal{T}_{\mathrm{tot}}$.
\end{theorem}

\subsection{Single-control QPE}

The $|\log(\epsilon)|$-qubit control register is not essential for a quantum speedup in quantum phase estimation.
This has been known since it's conception; Kitaev's original presentation of quantum phase estimation~\cite{kitaevQuantumMeasurementsAbelian1995} used only a single control qubit, a semi-classical quantum Fourier transform using only one control qubit was presented in Ref.~\cite{griffithsSemiclassical1996} around the same time, and though Shor's original algorithm~\cite{shorAlgorithmsQuantumComputation1994} used a control register, this was simplified shortly after~\cite{moscaHidden1999, zalkaFast1998, parkerEfficient2000, beauregardCircuit2003}.
These original algorithms relied on performing the QFT sequentially; reading bits of $\phi$ one-by-one via the single control qubit whilst partially projecting the system register into the corresponding eigenstate $|\phi\rangle$.
This is unnecessary when $|\phi\rangle$ is already prepared in an eigenstate, in which case the measurement does not affect the system register.
Discarding and re-preparing the system register separates the algorithm into classical post-processing of a series of single-qubit Hadamard tests, which can be advantageous in the presence of noise.
However, as an optimal probe state is no longer prepared, these single-control QPE algorithms do not obtain the Heisenberg limit~\cite{higginsEntanglementfree2007}.
(This is not an issue for NP algorithms such as factoring as one can classically confirm the result and repeat till success.)
This problem was solved up to a $\log^*$ factor by the maximum likelihood algorithm of Ref.~\cite{svore2013faster}, and later up to a constant factor by the robust phase estimation algorithm (RPE) of Ref.~\cite{kimmelRobust2015}, and the Bayesian phase estimation algorithm of Ref.~\cite{wiebeEfficient2016}.
The RPE algorithm, which repeats lower-order Hadamard tests multiple times to increase the confidence in more significant bits of the phase $\phi$ was further optimized Ref.~\cite{russo2021evaluating} and Ref.~\cite{belliardoAchieving2020}.
The latter paper provided an analytic optimization of the hyperparameters of the RPE algorithm that we will use in this work, proving a separation of around a factor $25$ from the strict Heisenberg limit $\mathcal{T}_{\text{tot}}=\pi/\epsilon$.
Indeed, this limit cannot be achieved with only a single control qubit; it was shown in Ref.~\cite{najafiOptimum2023} that this requires two control qubits (and can be achieved with such).

The replacement of the control register by a single qubit, and the circuit depth reduction from going to a single qubit has generated much interest in single-control QPE methods for NISQ or early-fault-tolerant research.
This was made stronger by the demonstration that one can remove control qubits in some cases~\cite{luAlgorithms2021}, which leads to a natural error mitigation strategy via verification of the initial state~\cite{obrienError2021}.
However, the restriction to start with an eigenstate is impractical for quantum simulation applications, unless one allows access to state preparation methods such as those in Refs.~\cite{geFaster2019,linNearoptimal2020} (which require additional control registers).
In Ref.~\cite{obrien2019quantum}, one of the authors demonstrated that a) single-control methods do not require an initial state, and b) that one has a freedom of choice to trade between shorter QPE circuits and more repetitions (with a quadratic overhead as one is reduced to the sampling-noise limit).
This freedom of choice is ultimately bounded by the gap between eigenenergies, which was shown earlier in Ref.~\cite{moitraSuperresolution2015}; the problem of phase estimation of a continuous (or near-continuous) spectrum was formalised as the quantum eigenvalue estimation problem in Ref.~\cite{sommaQuantum2019}.
Further work based on Bayesian~\cite{bergEfficient2021} and integral transform~\cite{roggeroSpectral2020} among other methods appeared, however the Heisenberg limit was not achieved with single-control methods without eigenstate access till Ref.~\cite{linHeisenbergLimited2022} and separately in Ref.~\cite{dutkiewicz2022heisenberg}.
This was optimized further in Refs.~\cite{dingEven2023, dingSimultaneous2023, wangQuantum2023}.

For future use, we now state the RPE algorithm as taken from Ref.~\cite{belliardoAchieving2020}.

\begin{figure}[ht]
    \label{fig:hadamard_test}
    \begin{equation*}
    \mbox{
        \Qcircuit @C=1em @R=0.1em { 
            \lstick{\ket{0}} & \gate{H}&\ctrl{2} &\gate{H}& \meter\\ 
            \push{\rule{0em}{2em}}&&&&&&&&&&&\\
             \lstick{\text{\Large$\ket{\phi}$}}&{/}\qw& \gate{\text{\huge $U^k$}} &\qw{/} & \qw\\
        }
    }
    \mbox{
        \Qcircuit @C=1em @R=0.1em { 
            \lstick{\ket{0}} & \gate{H}&\ctrl{2} &\gate{S^\dagger}&\gate{H}& \meter\\ 
            \push{\rule{0em}{2em}}&&&&&&&&&&&\\
             \lstick{\text{\Large$\ket{\phi}$}}&{/}\qw& \gate{\text{\huge $U^k$}} &\qw{/} & \qw&\qw\\
        }
    }
    \end{equation*}
    \caption{Hadamard test circuits with exponent $k$.
    Fist, a $\ket{+}$ state is prepared on the control qubit, and state $\ket{\phi}$ is prepared on the system register. Then, ${cU}$ is applied $k$ times. Finally, control qubit is measured in the X basis to obtain a single bit $x$. The procedure is repeated with for measurement in the Y basis, to obtain another bit $y$.
    The output is defined as $Z = (-1)^x + i(-1)^y$.
    The oracular cost is $2k$.
    This circuit produces a random variable $Z \in \{\pm 1 \pm i\}$ with probability distribution}
    \begin{equation}
    \label{eq:rpe_prob}
        P^{(HT)}((-1)^x  + i (-1)^y | \phi) = \frac{1+(-1)^x \cos(k\phi)}{2}\frac{1+(-1)^y \sin(k\phi)}{2}.
    \end{equation}
\end{figure}
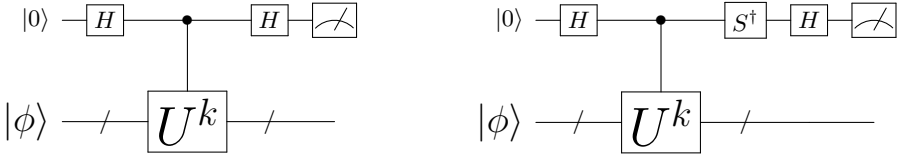

\begin{algorithm}
\label{alg:parametrised_RPE}
    [Robust Phase Estimation (RPE)]
    Input: number of orders $J$, vector of numbers of samples $\vec{M} \in \NN^{J}$.
    \item For $j =0, 1, 2, 3, ..., J-1$:
        \begin{enumerate}
            \item Fix $k = 2^j$.
            \item Run $M_j$ repetitions of the circuits of Fig.~\ref{fig:hadamard_test} with exponent $k$ to obtain a set of bits $\{Z_{i}\}$.
            \item
            Calculate the average $\bar{Z}$ of measurement results $\{Z_{i}\}$.
            \item Compute
            $\widetilde{\theta}^{(j)}=\mathrm{Arg}[\bar{Z}]\in[0,2\pi)$.
            \label{step:singlephase_deftheta}
            \item If $j=0$, set $\widetilde{\phi}^{(0)}=\widetilde{\theta}^{(0)}$.
            \item \label{step:update-QPE} Else, set $\widetilde{\phi}^{(d)}$ to be the unique value in the interval  $[\widetilde{\phi}^{(j-1)}-\frac{\pi}{k},\widetilde{\phi}^{(j-1)}+\frac{\pi}{k})$ (with periodic boundaries) such that 
            \begin{equation}
            k\widetilde{\phi}^{(j)} = \widetilde{\theta}^{(j)}\mod 2\pi.
            \end{equation}
        \end{enumerate}
        \item Return $\widetilde{\phi}=\widetilde{\phi}^{(J-1)}$.
\end{algorithm}

In the above algorithm we do not fix the values of the maximum order $J$, nor the number of repeat measurements at each order $M_j$.
This is because these numbers will change in the presence of global depolarizing noise in the next section.
However, in the noiseless setting, near-optimal choices of $M_j$ and $J$ are known thanks to Ref.~\cite{kimmelRobust2015}.
\begin{algorithm}
\label{alg:RPE_noiseless}
    [RPE noiseless hyperparameters]
    Input: target precision $\Delta$, metaparameters $\alpha, \beta$.
    Run Alg.~\ref{alg:parametrised_RPE} with $J = \lceil\log_2 (1/\Delta)\rceil$, and $M_j=\alpha(J-j-1) + \beta$.
\end{algorithm}
It was proven in Ref.~\cite{kimmelRobust2015} that for a range of $\alpha$, $\beta$, RPE achieves the Heisenberg limit up to a constant multiplicative factor; $\mathcal{T}_{\mathrm{tot}}\leq c\epsilon^{-1}$.
Near-optimal values for these constant factors were found and bounds tightened in Ref.~\cite{belliardoAchieving2020} yielding the following result:
\begin{theorem}\cite{belliardoAchieving2020}
\label{thm:rpe_noiseless}
    Fix target precision $\epsilon_\mathrm{t} > 0$. With input parameters $\Delta=0.409\epsilon_\mathrm{t}$, $\alpha= 4.0835$, $\beta=11$, RPE solves Def.~\ref{def:noiseless_QPE} for $\epsilon_\mathrm{t}=\epsilon$ with oracular cost $\mathcal{T}_{\mathrm{tot}}\leq 24.26\pi\epsilon^{-1}$ and maximal depth $\mathcal{T}_{\text{max}} \leq 2\Delta^{-1} \leq 5\epsilon ^{-1}$.
\end{theorem}
Note again the distinction between the input parameter $\epsilon_\mathrm{t}$ and noise bound $\epsilon$ that was made in Alg.~\ref{alg:sine-state-single-circuit}.
The above result is an analytic bound that is not perfectly tight; numerically we find that with the chosen parameters RPE achieves error $\epsilon\approx 5\pi/\mathcal{T}_{\mathrm{tot}}$.

\subsection{Information theoretic bounds}\label{sec:fisher_information}

To benchmark the algorithms for Problem~\ref{def:problem_QPE}, it is useful to know what lower bounds exist.
Every phase estimation algorithm has some distribution $P^{(\mathrm{alg})}(\widetilde\phi|\phi)$ of outputs $\widetilde\phi$ depending on the true phase $\phi$.
We are interested in the average Holevo error
\begin{equation}
    \epsilon_H^{(\mathrm{alg})}(\mathcal{T}_{\mathrm{tot}}) = \sqrt{\int_0^{2\phi}\frac{d\phi}{2\pi}\int_0^{2\pi}d\widetilde\phi P^{(alg)}(\widetilde\phi|\phi) 4\sin^2\bigg(\frac{\widetilde\phi-\phi}{2}\bigg)}.
\end{equation}
We want to optimize this and find the minimum error that can be achieved at a fixed cost $\mathcal{T}_{\mathrm{tot}}$, but this is not practical as we do not have a simple form of $P^{(\mathrm{alg})}$ for an arbitrary algorithm.
Instead, one can fix a choice of quantum circuits (e.g. repetitions of SinQPE with different control dimensions $K$, or repetitions of single Hadamard tests on $U^k$ for different $k$), and bound the performance of any classical algorithm on the classical output distributions $\vec{x}$ of these circuits.
The Cramer-Rao theorem~\cite{cramerMathematical1946,raoInformation1992}, bounds the Holevo from below by the inverse of the Fisher information
\begin{equation}
    \epsilon_H^2 \geq \frac{1}{\mathcal{I}(\phi)}, \qquad \mathcal{I}(\phi)= \int dx P(\vec{x}|\phi)\Bigg[\frac{d\log P(\vec{x}|\phi)}{d\phi}\Bigg]^2.
\end{equation}
For the Hadamard test, it turns out that $\mathcal{I}(\phi)$ is independent of $\phi$~\cite{dutkiewicz2022heisenberg}, however this is not true for SinQPE.
To compensate for this, it is common to take the Bayesian Cramer-Rao bound~\cite{cramerMathematical1946,casellaStatistical2002}, where the Fisher Information is replaced by the average over all possible phases
\begin{equation}
\label{eq:fisher_info}
    \epsilon_H^2 \geq \frac{1}{\bar{\mathcal{I}}}, \qquad \bar{\mathcal{I}} = \int_0^{2\phi}\frac{d\phi}{2\pi} \mathcal{I}(\phi)
\end{equation}
We will invoke this definition and the Cramer-Rao bound for fixed probability distributions later in this work, in order to bound the error in a QPE algorithm from above.
As the Fisher information is additive, we can use invoke these results on a single circuit at a time and add their results.
We will find that for all our considered algorithms, the lowest cost is achieved by running an increasing number $M$ of circuits with depth $\mathcal{T}_*$ that maximises the specific Fisher Information $\bar{\mathcal{I}}(\mathcal{T})/\mathcal{T}$ \cite{dutkiewicz2022heisenberg, kolodynski2013efficient},
\begin{equation}
\label{eq:specific_fi_lower_bound}
    \epsilon_H^2 \leq \frac{1}{\bar{\mathcal{I}}(T) M} = \frac{T}{\bar{\mathcal{I}}(\mathcal{T}) \mathcal{T}_{\mathrm{tot}}} \leq \frac{\mathcal{T}_*}{\bar{\mathcal{I}}(\mathcal{T}_*) \mathcal{T}_{\mathrm{tot}}}.
\end{equation}

\topic{quantum lower bounds, Heisenberg limit, GDN limit}
In this work we fix our quantum circuits and optimise over the parameters of the probability distributions, but in some cases it is possible to optimize the Fisher information over all choices of measurements on all states generated using a fixed oracle.
This yields the channel Fisher information, which characterises the black box in which the parameter $\phi$ is encoded.
For the noiseless version of QPE, Problem \ref{def:noiseless_QPE}, the channel Fisher information can be taken to yield the ultimate bound on the error: the \emph{Heisenberg limit} $ \epsilon \geq \frac{\pi}{\mathcal{T}_{\mathrm{tot}}}$ \cite{giovannettiQuantumenhanced2004}.
This is saturated by the probability distribution $P^{(SinQPE)}$ (Eq.~\ref{eq:sinqpe_noisless_prob_distribution})~\cite{vandamOptimal2007}.
For a noisy oracle subject to global depolarizing noise with rate $\gamma$ (see App.~\ref{sec:qpe-global-depol}), the channel Fisher information can be bounded above, yielding $\epsilon \geq \sqrt{\frac{\gamma}{\mathcal{T}_{\mathrm{tot}}}}$ \cite{dutkiewicz2022heisenberg}.

\section{Quantum phase estimation with global depolarizing noise}\label{sec:qpe-global-depol}

Hardware noise in quantum devices can take many forms, and some knowledge of the noise model is often required in order to perform error mitigation. 
In this appendix we focus on mitigating a simple, well characterized noise model in a QPE experiment; global depolarizing noise with a fixed strength $\gamma$.
This will allow us to more rigorously analyse and compare the performance of different estimators; we will consider the more general form in App.~\ref{sec:unbiasing_qpe}.

Global depolarizing noise (GDN)~\cite{kingCapacity2002,nielsen2001quantum,obrien2019quantum,vovroshSimple2021,hugginsVirtual2021} is a stochastic noise model where with some probability $p$ a random global unitary is applied to the system, replacing the system's state with the maximally mixed state.
This maps our controlled unitary $U_{\mathrm{c}}$ to the quantum channel $\widetilde{\mathcal{U}}_\gamma^{(GDN)}$, defined by
\begin{equation}\label{eq:GDN_def}
    \widetilde{\mathcal{U}}_\gamma^{(GDN)}[\rho] = e^{-\gamma} U_\text{c} \rho U_\text{c}^\dag + (1-e^{-\gamma}) \rho_\mathbb{1}.
\end{equation}
This definition is exact as global depolarizing noise process commutes with any subsequent unitary operations.
This allows us to calculate the effect of repeated GDN channels interleaved by arbitrary operations as
\begin{equation}\label{eq:GDN_k}
    {(\widetilde{\mathcal{U}}^{(GDN)}_\gamma)}^{k}[\rho] = 
    e^{-\gamma\,k} (U_\text{c})^k \rho (U_\text{c}^\dag)^k + (1-e^{-\gamma\,k}) \rho_\mathbb{1} = \widetilde{\mathcal{U}}^{(GDN)}_{k\gamma}[\rho]
\end{equation}
We see that this yields an exponential decay in circuit fidelity regardless of the unitary; with probability $1-e^{-\gamma k}$ the system yields the maximally-mixed state, and all useful information about the target problem (e.g. in our case information about the phase $\phi$) is lost.
This makes global depolarizing noise a particularly simple noise model to analyse, though it comes with an exponential decay in fidelity that cannot be mitigated or even error corrected without an overhead that is exponential in $k$~\cite{caiQuantum2023}.
This is in contrast to E.g.~\emph{local} i.i.d. depolarizing noise on individual qubits.
These do not commute with unitary operations, and so a simple form like Eq.~\eqref{eq:GDN_k} cannot be found, but this noise model allows error correction when error rates are below the fault-tolerant threshold~\cite{fowlerLow2019}.

In this appendix we optimize and compare two different variants of QPE robust to global depolarizing noise.
One of these will be based on the optimized RPE algorithm of Ref.~\cite{belliardoAchieving2020}, while for the other we develop a robust version of sin-state QPE, using maximum-likelihood estimation.
The main outcome of this analysis is a black box that takes in an error $\epsilon$ and fixed noise rate $\gamma$, and tells the number of calls to the oracle $\mathcal{T}_\text{tot}$ that would be required to solve
\begin{definition}[QPE in the presence of depolarising noise]
    The problem in Definition~\ref{def:problem_QPE} in the presence of global depolarizing noise with a known, fixed error rate $\gamma$ per application of the unitary oracle, i.e.~$\widetilde{\mathcal{U}} = {\widetilde{\mathcal{U}}^{(GDN)}_\gamma}$.
\end{definition} 
This will be later used in App.~\ref{sec:compilation} to optimize surface code parameters according to the framework defined in Fig.~\ref{fig:design_principles}.

\subsection{Single-control QPE in the presence of noise}

If we directly apply RPE (Alg.~\ref{alg:RPE_noiseless}) in the presence of depolarising noise, the (Holevo) error $\epsilon_H$ in the estimate will plateau as the target precision $\epsilon_\mathrm{t}$ is decreased, in spite of the growing cost.
This is because the algorithm requires running ever-deeper circuits, whose output converges exponentially quickly to uniform noise.
In Ref.~\cite{kimmelRobust2015}, it was shown that the RPE algorithm is robust to the probabilities $p_j(Z)$ being corrupted by an arbitrary noise term of strength no greater than $\frac{1}{\sqrt{8}}$.
For global depolarizing noise this holds for all distributions $p_j$ till $j=j_{\text{max}}=\lfloor \log_2[\frac{\ln(2/\sqrt{8})}{\gamma}]\rfloor$, which estimates $\phi$ with high confidence to error $\epsilon \leq \frac{2\pi}{2^j}\sim\mathcal{O}(\gamma)$.
The algorithm will work with reduced confidence slightly beyond here, but will quickly start to fail.
To adapt the algorithm to the presence of noise, we can instead stop at some maximum depth around $\mathcal{T}_{\text{max}} \sim 2^{j_{\text{max}}}$, and repeat these last circuits many times to increase precision at the sampling noise limit (as was suggested in Refs.~\cite{kolodynski2013efficient,belliardoAchieving2020,dutkiewicz2022heisenberg}).
To adapt RPE to GDN we identify two distinct regimes: firstly when $\epsilon \gg \gamma$, in which case the number of samples are rescaled to ensure the final error is within the required precision, and secondly when $\epsilon \lessapprox \gamma$, in which case the maximum depth is fixed to some $\mathcal{T}_*$, and most samples are taken at this maximum depth.

As discussed in section \ref{sec:fisher_information}, the optimal maximal depth $\mathcal{T}_*$ can be quantified by specific Fisher Information \cite{dutkiewicz2022heisenberg, kolodynski2013efficient}. 
In the case of the Hadamard test with global depolarising noise, the Fisher Information $\bar{\mathcal{I}}_\gamma(\mathcal{T}) = e^{-2\gamma \mathcal{T}} \mathcal{T}^2$ can easily be calculated by substituting the probability distribution from Eq.~\eqref{eq:rpe_prob} into Eq.~\eqref{eq:fisher_info}. 
The depth that maximises $\bar{\mathcal{I}}_\gamma(\mathcal{T})/\mathcal{T}$ is $\mathcal{T}_* = 1/2\gamma$ \cite{dutkiewicz2022heisenberg}.
Asymptotically, the best performance we can get is given by Eq.~\eqref{eq:specific_fi_lower_bound}
\begin{equation}
\label{eq:noisy rpe_lower_bound}
    \mathcal{T}_{\mathrm{tot}} \geq \frac{2\mathcal{T}_*}{\bar{\mathcal{I}}_\gamma(\mathcal{T}_*) \epsilon^2} = \frac{4e\gamma}{\epsilon^2}.
\end{equation}
In practice, this bound will not be saturated, as we need to allocate a number of shots for circuits with lower depths in order to determine the most-significant bits of the phase with sufficiently high confidence.
In Ref.~\cite{belliardoAchieving2020}, the authors 
optimize the number of measurements $M_j$ of Alg.~\ref{alg:parametrised_RPE} to account for global depolarizing noise.
This results in a well-defined algorithm, whose cost scales asymptotically as $\mathcal{T}_\text{tot} \sim \mathcal{O}(\epsilon^{-2} / \gamma)$ (we will later fit the prefactor numerically):
\begin{algorithm}
\label{alg:RPE_noisy}
    [RPE with global depolarizing noise \cite{belliardoAchieving2020}]
    Fix a target precision $\epsilon_\mathrm{t}$, and define the constants
    \begin{align}
        J = \lfloor\log_2(1/\epsilon_\mathrm{t})\rfloor, 
        \quad \beta = 11 
        &\quad \text{if} \quad \epsilon_\mathrm{t} > \gamma
        \\
        J = \lfloor\log_2(1/\gamma)\rfloor, 
        \quad \beta = 11\frac{\gamma^2}{\epsilon_\mathrm{t}^2}
        &\quad \text{if} \quad \epsilon_\mathrm{t} < \gamma
    \end{align}

    Run Alg.~\ref{alg:parametrised_RPE} with $J$ orders  and numbers of samples
    $M_j = e^{2\gamma k_j}(\alpha(J-j) + C\gamma(2^J - 2^j) + \beta)$, with $\alpha = 4.0835$, $C = 1.3612$.    
\end{algorithm}

Several alternative single-control QPE algorithms to RPE have been proposed for the early-FT regime \cite{dingEven2023, linHeisenbergLimited2022, wangQuantum2023} and studied in the presence of GDN \cite{ding2023robust, liangModeling2024}.
In Ref.~\cite{nelsonAssessment2024} the costs of surface-code implementations of several algorithms were analysed, comparing the performance of various estimators under GDN.
The study found that when the initial state is an eigenstate, as in our case, the total runtime is minimized by either RPE or the quantum complex exponential least-squares algorithm of Ref.~\cite{dingEven2023}. 
For simplicity, we choose RPE with the parameters of Ref.~\cite{belliardoAchieving2020} as the benchmark to compare to.

\subsection{Multi-circuit QFT-based QPE}\label{sec:sin_state_qpe_multi_circuit}

Similar to the single-control case, in the presence of noise a sin-state QPE circuit will lose information as it grows arbitrarily deep.
This effect is even more drastic than for RPE, because a single measurement will be used; if we try to apply Algorithm~\ref{alg:sine-state-single-circuit} using a noisy oracle, the actual error $\epsilon$ will become larger for decreasing target error $\epsilon_\mathrm{t}$ (instead of plateauing).
This can be seen from Eq.~\eqref{eq:GDN_def}; the final measurement will return an estimate with the same error as the noiseless algorithm with probability $e^{-\gamma \mathcal{T}}$, and a fully random estimate with error $\mathcal{O}(1)$ with probability $1-e^{-\gamma \mathcal{T}}$, yielding a Holevo variance
\begin{equation}
   \label{eq:holevo_error_with_gdn} \epsilon_H^2(\gamma, \mathcal{T}) = e^{-\gamma \mathcal{T}} \tan^2\left(\frac{\pi}{\mathcal{T}+2}\right) + 2 \, (1-e^{-\gamma \mathcal{T}}).
\end{equation}
The above has a minimum $\min_\mathcal{T} \epsilon(\gamma, \mathcal{T}) = \sqrt{3}\sqrt[3]{\pi\gamma} + \mathcal{O}(\gamma) \approx 2.54 \gamma^{1/3} > 0$ at $\mathcal{T} = \sqrt[3]{\pi^2/\gamma} + \mathcal{O}(\gamma^{1/3})$ for any $\gamma >0$.
Indeed, achieving arbitrary precision using a single QPE circuit (as in Fig.~\ref{fig:sinqpe}) is impossible in the presence of global depolarising noise.
Moreover, naively averaging the outcomes over many circuit runs would result in $\mathcal{T}_{\mathrm{tot}} \propto \gamma^{1/3}\epsilon^{-2}$, as compared to the scaling  $\mathcal{T}_{\mathrm{tot}} \propto \gamma\epsilon^{-2}$ in the previous section for RPE.
We would hope to recover this scaling.

As in the previous section, to adapt Algorithm~\ref{alg:sine-state-single-circuit} to the presence of noise, we want to limit the maximum depth of a circuit and instead run multiple circuits.
This requires both defining the best choice of circuits to run, and determining a method to combine the results of these circuits to optimize the final precision.
We can write down the probability distribution of sin-state QPE in the presence of global depolarizing noise with strength $\gamma$ as 
\begin{equation}
    P_{\gamma}^{(\mathrm{SinQPE})}(x|\phi)=e^{-\mathcal{T}\gamma}P^{(\mathrm{SinQPE})}(x|\phi) + (1-e^{-{\mathcal{T}}\gamma})\frac{1}{\mathcal{T}+1}
\end{equation}
where $P^{(\mathrm{SinQPE})}$ is the distribution of the noiseless circuit (Eq~\ref{eq:sinqpe_noisless_prob_distribution}).
An asymptotically optimal estimator is then to perform maximum-likelihood estimation on the joint noisy distributions $P_{\gamma}^{(\mathrm{SinQPE})}(x|\phi)$ taken from different runs.
(This is slightly different, but related to the maximum-likelihood estimation performed on single-control QPE data by Ref.~\cite{svore2013faster, dingEven2023}.)
The maximum-likelihood routine can be performed classically with classical cost $\mathcal{O}(\epsilon^{-2})$, which we assume is a negligible overhead to the total cost of the quantum algorithm $\mathcal{T}_{\mathrm{tot}}\sim\mathcal{O}(\epsilon^{-2})$.
We numerically find that for the sin-state circuit with GDN, $\argmax_\mathcal{T} \bar{\mathcal{I}}(\mathcal{T})/\mathcal{T} \approx \gamma^{-1}$.
It remains to choose what to do for larger target errors.

\subsubsection{Algorithm details}\label{sec:algorithm_details}

There are 3 regimes characterised by the target precision $\epsilon$ relative to the noise strength $\gamma$.
When $\epsilon \gg \gamma$, the effects of noise are negligible and the optimal strategy is the same as for the noiseless case - run a single circuit of sufficient depth, such that the variance of the output is $\leq \epsilon^2$.
As the depth increases, the output becomes more and more influenced by noise till a maximal depth $\mathcal{T}_{\text{max}}\sim\mathcal{O}(\gamma^{-1})$ beyond which running larger depths is never efficient.
Instead, asymptotically small errors can be achieved by repeating this single circuit multiple times.
In this regime the best possible precision can be achieved by processing the circuit samples through maximum likelihood estimation.
Maximum likelihood estimation is only optimal in the limit of a large number of samples; as a consequence, in the intermediate regime ($\gamma \approx \epsilon$), we will need to run a large enough number of samples for circuits of depth smaller than for the asymptotic regime, effectively interpolating circuit depth and number of shots between the two edge cases.

To optimize our measurement strategy for the intermediate regime, we first choose a maximal depth $\mathcal{T}_1$ for any single experiment.
With a small number of samples, the error is characterised by the variance of the circuit output probability distribution.
We need to find $\mathcal{T}_1$ such that when the total resources are $\mathcal{T}_1$ the expected error is smaller when running a circuit with depth $\mathcal{T}_1$ once, then if the circuit of depth $\mathcal{T}_1/2$ is repeated twice.
In the first case, the expected error is simply $\epsilon(\mathcal{T}_1,\gamma)$.
In the second case, we get two independent samples, each coming from the target distribution with probability $p = e^{-\gamma \mathcal{T}_1/2}$.
There are three possibilities to consider for this second case: if both samples come from the target distribution, the output estimate $\widetilde\phi$ is their average, and the error is $\approx \frac{1}{\sqrt 2} \frac{\pi}{T1/2}$; if one of the samples comes from the target distribution, the error is either $\frac{\pi}{\mathcal{T}_1 / 2}$ or $2$; if both samples come from the noise distribution, the output is random and the error is $2$.
Summing over these possibilities for the second case, we find the threshold depth $\mathcal{T}_1$ by solving
\begin{equation}
   p^2 \frac{\pi^2}{\mathcal{T}_1^2} + (1-p^2)\times 2 \leq p^2 \times \frac{4\pi^2}{2\mathcal{T}_1^2} + 2p(1-p)\times\big(\frac{1}{2}\times \frac{4\pi^2}{\mathcal{T}_1^2}+ \frac{1}{2}\times 2\big) + (1-p)^2 \times 2.
\end{equation}
One can check this condition is equivalent to
\begin{equation} 
    \frac{\pi^2}{\mathcal{T}_1^2}  \leq
    2\frac{1-p}{1+3p/4} 
    =8\frac{1-e^{-\gamma \mathcal{T}_1/2}}{4-3e^{-\gamma \mathcal{T}_1/2}} =2\gamma \mathcal{T}_1 + O(\gamma^2\mathcal{T}_1^2)
\end{equation}
To first order in $\gamma \mathcal{T}_1$, we get that $\mathcal{T}_1 = \sqrt[3] {\frac{\pi^2}{2\gamma}} \approx 1.70\gamma^{-1/3}$ (which validates the above expansion for small $\gamma$).
The corresponding error is $\epsilon(\mathcal{T}_1;\gamma)\approx 1.84 \gamma^{1/3}$.

We have not yet chosen the depth $\mathcal{T}_{\text{max}}$ for the small error regime.
Again, it is best to choose the depth $\mathcal{T}_{\text{max}} = \mathcal{T}_2$ that maximizes the specific Fisher Information $\mathcal{T}_2 = \argmax_{T}\mathcal{I(T;\gamma)}/T$. Numerically we find $\mathcal{T}_2 = \gamma^{-1}$. In the limit of many samples $M$ (in practice, we assume the estimator has converged for $M>M_2 = 100$), the variance of the MLE is $(\bar{\mathcal{I}}M)^{-1}$ so to achieve $\epsilon \ll \gamma$ the number of samples required is $M = 1/\bar{\bar{\mathcal{I}}}_\gamma(\mathcal{T}_2)\epsilon^{2}$ leading to the total cost $\mathcal{T}_{\mathrm{tot}} = \mathcal{T}_2 M = \mathcal{T}_2/\bar{\mathcal{I}}_\gamma(\mathcal{T}_2)\epsilon^{2}$.

It remains to define the algorithm in the intermediate regime to interpolate between depths $\mathcal{T}_1$ and $\mathcal{T}_2$, and numbers of samples $1$ and $M_2$.
This is not easy, because while we know the expected asymptotic performance of the MLE it is hard to predict what happens when there are few samples.
Instead, we construct a model for the error $\epsilon(\mathcal{T}, M, \gamma)$, and then choose $\mathcal{T}$ and $M$ so that the total cost $\mathcal{T}_{\mathrm{tot}} = \mathcal{T}M$ is minimal, while $\epsilon(\mathcal{T}, M, \gamma) \leq \epsilon$.
We write the estimator error as a sum of two cases:
\begin{equation}
\label{eq:mle_error_model}
    \epsilon^2 \approx p_F\times 2 + (1-p_F)\frac{1}{\bar{\mathcal{I}}M}.
\end{equation}
That is, with probability $p_F$ the estimator fails, the outcome is completely random, and the average square error is $\int_0^{2\pi}\frac{d\widetilde\phi}{2\pi}4\sin^2(\frac{\widetilde\phi-\phi}{2})^2 = 2$, but with probability $1-p_F$ the estimator succeeds, and the average square error satisfies the Cramer-Rao bound $\frac{1}{\bar{\mathcal{I}} M}$.
The failure probability $p_F$ depends on the number of samples $M$, and the circuit fidelity $F = e^{-\gamma T}$.
We conjecture that the probability of failure decreases exponentially with the number of samples $p_F = e^{-\alpha M}$ (motivated by the asymptotic normality of the maximum likelihood estimator).
Furthermore, we choose $\alpha = \gamma T / 2$, so that $p_F(M=2) = e^{-\gamma \mathcal{T}}$ -- we expect this behavior because, when we have only two samples, there is no way to tell which of the two is more likely to be produced by the signal or by the noise; in that case MLE is forced to randomly pick one, and the estimator for two samples reduces to that with one sample, and the probability of success is equal to the fidelity $e^{-\gamma \mathcal{T}}$.
We assume that the probability of failure becomes negligible when the number of shots exceeds $M_2 = 100$.

As we want to optimise the performance of the algorithm for relatively small oracle depths,
it is worth noting that that we are not restricted to powers of 2.
In Fig.~\ref{fig:sinqpecircuit_arb_dim} we show how to generalise the SinQPE circuit of Fig~\ref{fig:sinqpecircuit} to a general control dimension $K\in\NN$.

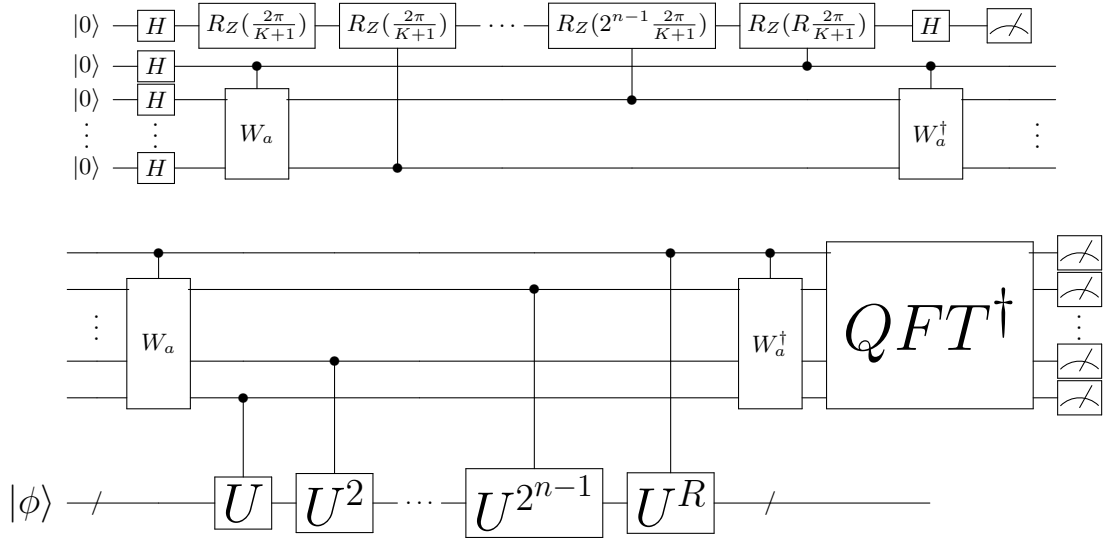
\begin{figure}[ht]
    \begin{equation*}
    \Qcircuit @C=1em @R=.1em {
     \lstick{\ket{0}} &\gate{H}&\gate{R_Z(\frac{2\pi}{K+1})}&\gate{R_Z(\frac{2\pi}{K+1})}&\push{\rule{0.15em}{0em}\dots\rule{0.15em}{0em}}\qw &\gate{R_Z(2^{n-1}\frac{2\pi}{K+1})}&\gate{R_Z(R\frac{2\pi}{K+1})}& \gate{H}&\meter \\
    \lstick{\ket{0}} & \gate{H}&\ctrl{1}&\qw &\qw&\qw&\ctrl{-1}&\ctrl{1}&\qw&\qw\\ 
    \lstick{\ket{0}} &\gate{H}&\multigate{2}{W_a}&\qw &\qw&\ctrl{-2}&\qw&\multigate{2}{W_a^{\dagger}}&\qw&\qw\\
    \lstick{\rotatebox{90}{$\cdots$}\ }&\push{\rotatebox{90}{$\cdots$}}&&&&&&&&\lstick{\rotatebox{90}{$\cdots$}}\\
    \lstick{\ket{0}} & \gate{H}&\ghost{W_R}&\ctrl{-4}&\qw&\qw &\qw&\ghost{W_R}&\qw&\qw\\ 
    }
    \end{equation*}
    
    \begin{equation*}
    \Qcircuit @C=1em @R=.1em {
    &\qw &\ctrl{1}&\qw&\qw&\qw&\qw& \ctrl{6}&\ctrl{1} & \multigate{4}{\text{\Huge $QFT^\dagger$}}& \meter\\ 
    &\qw&\multigate{3}{W_a}&\qw&\qw&\qw& \ctrl{5}&\qw&\multigate{3}{W_a^\dagger} & \ghost{\text{ \Huge $QFT^\dagger$}}& \meter\\
    &\push{\rotatebox{90}{$\cdots$}} & \pureghost{W_R}&&&&&& \pureghost{W_R}& \pureghost{\text{ \Huge $QFT^\dagger$}}&\rotatebox{90}{$\cdots$}\\
    &\qw&\ghost{W_R^\dagger}&\qw&\ctrl{3}&\qw&\qw&\qw &\ghost{W_R}& \ghost{\text{ \Huge $QFT^\dagger$}}& \meter\\ 
    &\qw&\ghost{W_R^\dagger}&\ctrl{2}&\qw&\qw&\qw &\qw&\ghost{W_R}& \ghost{\text{ \Huge $QFT^\dagger$}}& \meter\\ 
    \push{\rule{0em}{2em}}&&&\\
    \lstick{\text{\Large$\ket{\phi}$}}&{/}\qw&\qw& \gate{\text{\huge $U$}} & \gate{\text{\huge $U^2$}} & \push{\rule{0.15em}{0em}\dots\rule{0.15em}{0em}}\qw & \gate{\text{\huge $U^{2^{n-1}}$}} & \gate{\text{\huge $U^{R}$}} & \qw{/}&\qw\\
    }
    \end{equation*}
\caption{Implementation of sin-state QPE circuit for $K\in \NN$.
Let $n = \lfloor \log_2(K) \rfloor$, $R = K-2^n$.
If $R = 0$ ($K$ is a power of 2), then implement the circuit from Fig.~\ref{fig:sinqpecircuit}.
Otherwise, the $c_K U$ operation can be realised using a control register of $n+1$ qubits in the following way: First, the most significant qubit controls an addition gate  $W_{a}=\sum_{j=0}^{2^{n-1}-a-1}\ket{j+a}\bra{j}$ for $a=2^n-R$ in the remaining register ($W_a$ can be constructed with up to $n$ Toffoli gates and as many auxiliary qubits \cite{gidneyHalving2018, sandersCompilation2020}), and then $R$ applications of $U_c$. Then, the remaining register controls a $c_{2^{n}}U$. Finally, the addition of $a$ is uncomputed.
We can apply the same trick to the controlled rotations in the preparation of the sin-state to prepare $\ket{s_K}$.
$QFT$ of order $K$ can be realised with cost $\widetilde{\mathcal{O}}(\log(K))$ \cite{hales2000improved,mosca2003exact}.
The oracular cost is $K-1$, as in the case of $K = 2^n$ (Fig.~\ref{fig:sinqpecircuit}).}
\label{fig:sinqpecircuit_arb_dim}
\end{figure}

\begin{algorithm}\label{alg:mle-sinqpe}
    [Multi-circuit sin-state QPE with GDN]
    Input: Target precision $\epsilon_\mathrm{t}$, noise strength $\gamma$.
    \begin{enumerate}
        \item Let the threshold parameters be  $\mathcal{T}_1 = \lfloor\sqrt[3]{2\pi^2/3\gamma}\rfloor$, $\epsilon_1 = \epsilon_H(\gamma, \mathcal{T}_1)$ according to Eq.~\ref{eq:holevo_error_with_gdn}, and $\mathcal{T}_2 = \lfloor\gamma^{-1}\rfloor$, $M_2 = 100$  $\epsilon_2 = 1/\sqrt{\bar{\mathcal{I}}_\gamma(\mathcal{T}_2) M_2}$.
        \item Choose the number of shots $M$ and circuit depth $\mathcal{T}$:
        \begin{enumerate}
            \item If $\epsilon_\mathrm{t} \geq \epsilon_1$: 
            $\mathcal{T} = \lceil\frac{\pi}{\arctan(\epsilon_\mathrm{t})}-2\rceil$, $M = 1$
            \item If $\epsilon_\mathrm{t} \leq \epsilon_2$: $\mathcal{T} = \mathcal{T}_2$, $M = 1/\bar{\mathcal{I}}(\mathcal{T}_2, \gamma)\epsilon_\mathrm{t}^2$.
            \item If $\epsilon_1 < \epsilon_\mathrm{t} < \epsilon_2$: $\mathcal{T}, M = \argmin_{\mathcal{T} \in \mathbb{N}, M \in \mathbb{N}}(\mathcal{T} \cdot M)$ under the constraint
            \begin{equation}
                2e^{-\gamma \mathcal{T}M/2} + (1-e^{-\gamma \mathcal{T}M/2})\frac{1}{\bar{\mathcal{I}}_\gamma(\mathcal{T}) M} \leq \epsilon_\mathrm{t}^2.
            \end{equation}
        \end{enumerate}
        \item Run $M$ shots of the SinQPE circuit ((Fig.~\ref{fig:sinqpecircuit_arb_dim})) with control dimension $K=\mathcal{T}+1$ to obtain $M$ samples $\{x_j\}$.
        \item Return $\argmin_{\widetilde\phi} L(\{x_j\}, \widetilde\phi)$.
    \end{enumerate}
\end{algorithm}
As before, we again define this algorithm in terms of a parameter $\epsilon_\mathrm{t}$ that allows us to fix the circuit parameters, but does not necessarily form a bound on the final Holevo error.
However, we find numerically that the true Holevo error bound $\epsilon$ lies very close to this target.

\subsubsection{Performance guarantees}
\label{sec:performance_guarantees}

Here we state and prove the formal version of Theorem \ref{thm:mle_sinqe_informal} in the main text.

{
\begin{theorem}
\label{thm:mle_sinqpe}
 Given an initial eigenstate preparation, the total number of uses $\mathcal{T}_{\mathrm{tot}}$ of the unitary required for Alg.~\ref{alg:mle-sinqpe} to produce a phase estimate with Holevo error $\epsilon \to 0$ in the presence of global depolarizing noise with strength $\gamma \to 0$ depends on the order in which these limits are taken. Specifically:

\begin{enumerate}
    \item $\lim_{\epsilon \rightarrow 0} \lim_{\gamma \rightarrow 0}\mathcal{T}_{\mathrm{tot}}\epsilon=\pi$ (i.e. $\gamma \ll \epsilon$)
    \item $\lim_{\gamma \rightarrow 0}\lim_{\epsilon \rightarrow 0} \mathcal{T}_{\mathrm{tot}}\epsilon^2/{\gamma}=C$ (i.e. $\epsilon \ll \gamma$), where $C\approx 20$ is a constant of proportionality.
\end{enumerate}
\end{theorem}
\begin{proof}
We start by proving statement 1., by noticing that for $\gamma \to 0$ we recover the noiseless case, and Algorithm~\ref{alg:mle-sinqpe} behaves identically to Algorithm~\ref{alg:sine-state-single-circuit}, as described in Theorem~\ref{thm:sinqpe-noiseless}.
Specifically, as $\gamma\to 0$ for a fixed $\epsilon$, Alg.~\ref{alg:mle-sinqpe} prescribes a single circuit with depth $\mathcal{T} = \lceil \frac{\pi}{\arctan(\epsilon)} - 2 \rceil$, and the total cost is $\mathcal{T}_{\mathrm{tot}} = \mathcal{T}$. The error $\epsilon$ is given by the Holevo variance of the noisy distribution, Eq.~\eqref{eq:holevo_error_with_gdn}.

    \begin{equation}
        \epsilon = \sqrt{e^{-\gamma \mathcal{T}_{\mathrm{tot}}} \tan^2\left(\frac{\pi}{\mathcal{T}_{\mathrm{tot}}+2}\right) + 2 \, (1-e^{-\gamma \mathcal{T}_{\mathrm{tot}}})}.
    \end{equation}
   Taking the limit $\gamma \to 0$, this becomes the same as in the noiseless case, $\epsilon \xrightarrow{\gamma\rightarrow 0}  \tan\left(\frac{\pi}{\mathcal{T}_{\mathrm{tot}}+2}\right) $.
   Thus the total cost satisfies
    \begin{equation}
       \lim_{\epsilon \to 0}\mathcal{T}_{\mathrm{tot}}\epsilon = \lim_{\mathcal{T}_{\mathrm{tot}}\to\infty} \mathcal{T}_{\mathrm{tot}}\times \tan\left(\frac{\pi}{\mathcal{T}_{\mathrm{tot}}+2}\right) = \pi.
    \end{equation}

    We turn to proving statement 2.
    When $\epsilon \to 0$ for a fixed noise rate $\gamma$, Alg.~\ref{alg:mle-sinqpe} uses a large number of samples $M = 1/\bar{\mathcal{I}}_\gamma(\mathcal{T}_2)\epsilon^2$ of circuit \ref{fig:sinqpecircuit} with depth $\mathcal{T}_2 = \lfloor\gamma^{-1}\rfloor$, so the total cost is $\mathcal{T}_{\mathrm{tot}} = M \mathcal{T}_2$.
    By the limiting properties of the maximum likelihood estimator (which saturates the Cramér-Rao bound for large number of samples $M$), the algorithm will converge to the target precision $\epsilon$, with the total cost satisfying
    \begin{equation}
        \lim_{\epsilon \rightarrow 0} \mathcal{T}_{\mathrm{tot}}\epsilon^2 = \frac{\mathcal{T}_2}{\bar{\mathcal{I}}_\gamma(\mathcal{T}_2)}.
    \end{equation}

    It remains to calculate the limit $\lim_{\gamma \rightarrow 0} \mathcal{T}_2/\gamma \bar{\mathcal{I}}_\gamma(\mathcal{T}_2)$. 
    For brevity of notation, let $f_\mathcal{T}(\Delta\phi)$ be such that the probability distribution in Eq.~\ref{eq:sinqpe_noisless_prob_distribution} can be written as 
    \begin{align}
    P_{0,\mathcal{T}}(x|\phi) &= f_{\mathcal{T}}\left(\phi - \frac{2\pi x}{\mathcal{T}+1}\right), \\
    f_{\mathcal{T}}(\Delta\phi) &= \frac{\sin^2\frac{\pi}{\mathcal{T}+2}}{(\mathcal{T}+1)(\mathcal{T}+2)} \frac{1+\cos[(\mathcal{T}+2)\Delta\phi]}{\left[\cos(\Delta\phi) - \cos(\frac{\pi}{\mathcal{T}+2})\right]^2}.
\end{align}
    The fidelity of the circuit depth $\mathcal{T}$ under GDN with strength $\gamma$ is $e^{-\gamma\mathcal{T}}$, and the associated probability distribution is $P_{\gamma,\mathcal{T}}(x|\phi) = e^{-\gamma \mathcal{T}}P_{0,\mathcal{T}}(x|\phi) + (1-e^{-\gamma \mathcal{T}})/(\mathcal{T}+1)$.
    The Fisher Information is by definition [Eq.~\eqref{eq:fisher_info}]
    \begin{align}
        \bar{\mathcal{I}}_\gamma(\mathcal{T}) &= \int_0^{2\pi}\frac{d\phi}{2\pi}
        \sum_{x=0}^{\mathcal{T}}
        \frac{e^{-2\gamma \mathcal{T}}f'^2_{\mathcal{T}}(\phi - 2\pi \frac{x}{\mathcal{T}+1})}{e^{-\gamma \mathcal{T}}f_{\mathcal{T}}(\phi - 2\pi \frac{x}{\mathcal{T}+1}) + (1-e^{-\gamma \mathcal{T}})/(\mathcal{T})}\\
        &= \sum_{x=0}^{\mathcal{T}}\int_{-\pi - 2\pi \frac{x}{\mathcal{T}+1}}^{\pi - 2\pi \frac{x}{\mathcal{T}+1}}\frac{ds}{2\pi}\frac{ e^{-2\gamma \mathcal{T}}f'^2_{\mathcal{T}}(s)}{e^{-\gamma \mathcal{T}}f_{\mathcal{T}}(s) + (1-e^{-\gamma \mathcal{T}})/(\mathcal{T}+1)}\\
        &= (\mathcal{T}+1) \int_{-\pi}^{\pi}\frac{ds}{2\pi}\frac{ e^{-2\gamma \mathcal{T}}f'^2_{\mathcal{T}}(s)}{e^{-\gamma \mathcal{T}}f_{\mathcal{T}}(s) + (1-e^{-\gamma \mathcal{T}})/(\mathcal{T}+1)}\\
        &= (\mathcal{T}+1)e^{-\gamma \mathcal{T}} \int_{-\pi}^{\pi}\frac{ds}{2\pi}\frac{ f'^2_{\mathcal{T}}(s)}{f_{\mathcal{T}}(s) + (e^{\gamma \mathcal{T}}-1)/(\mathcal{T}+1)},
    \end{align}
    where we have used the periodicity of $f_{\mathcal{T}}$ to remove the sum and dependence on $\phi$.
    We can substitute $s \mapsto (\mathcal{T}+2)s$ to rewrite this as
    \begin{equation}
        \bar{\mathcal{I}}_\gamma(\mathcal{T}) = e^{-\gamma \mathcal{T}} \times (\mathcal{T}+2)(\mathcal{T}+1) \times \int_{-(\mathcal{T}+2_)\pi}^{(\mathcal{T}+2)\pi}\frac{1}{2\pi} \frac{(\mathcal{T}+2)^{-2}f'^2_{\mathcal{T}}(\frac{s}{\mathcal{T}+2})ds}{f_{\mathcal{T}}(\frac{s}{\mathcal{T}+2}) + (e^{\gamma \mathcal{T}}-1)/(\mathcal{T}+1)}.
    \end{equation}
    Using L'Hôpital's rule, we can show that in the limit $\mathcal{T}\to\infty$, the function $f_\mathcal{T}$ converges pointwise to
    \begin{equation}
    \label{eq:limit_prob_func}
        f(s) = \lim_{\mathcal{T}\rightarrow \infty} f_\mathcal{T}(\frac{s}{\mathcal{T}+2}) = (2\pi)^2 \frac{1+\cos s}{(\pi^2 - s^2)^2},
    \end{equation}
    and the derivative converges as
    \begin{equation}
    \label{eq:limit_prob_deriv}
        \frac{1}{\mathcal{T}+2}f'_{\mathcal{T}}\left(\frac{s}{\mathcal{T}+2}\right) \xrightarrow{\mathcal{T}\to\infty} f'(s).
    \end{equation}
    Therefore, in the limit $\gamma \rightarrow 0$, where the fidelity is $e^{\gamma \mathcal{T}_2}, \rightarrow e$ the integrand converges pointwise to 
    \begin{align}
    \label{eq:limit-score}
        j(s) &= \frac{1}{2\pi} \frac{f'^2(s)}{f(s)}\\
        &=\frac{1}{(s^2 - \pi^2)^2}\left[(1-\cos s) + \frac{16s^2}{(s^2 - \pi^2)^2}(1+\cos s)+\frac{8s}{s^2 - \pi^2}\sin s\right],
    \end{align}
    which we can bound by bounding all the trigonometric functions as $|\sin s|,|\cos s| \leq 1$ as
    \begin{equation}
    \label{eq:limit-score-bound}
        |j(s)|
        \leq \frac{1}{(s^2 - \pi^2)^2}\left|2 + 2\frac{16s^2}{(s^2 - \pi^2)^2}+\frac{8s}{s^2 - \pi^2}\right|
    \end{equation}
    For large $s$, we can bound this as $|j(s)| = \mathcal{O}(|s|^{-4}) $, so the integral converges and
    \begin{equation}
       \frac{\gamma}{\mathcal{T}_2}\bar{\mathcal{I}}_\gamma(\mathcal{T}_2) \xrightarrow{\mathcal{T}\to\infty} e^{-1} \int_{-\infty}^{+\infty} j(s)ds \equiv C^{-1}.
    \end{equation}
    We numerically find the value of the integral is $e/C \approx 0.13$, and $C \approx 20.8$.
\end{proof}
}

\subsection{Comparison between single-control and QFT-based methods}
\label{sec:qpe_gdn_numerical_comparison}

\begin{figure}[h!]
    \centering
     \includegraphics{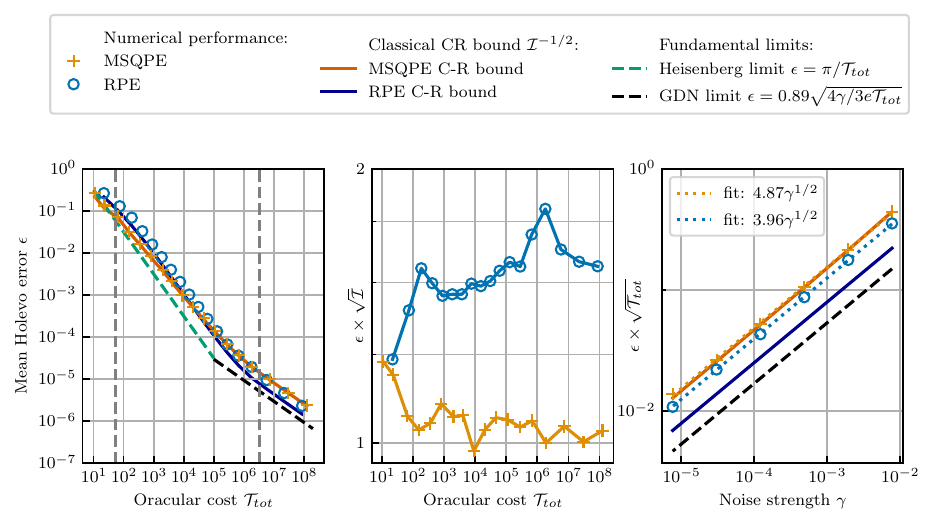}
    \caption{
        Comparison of RPE (Alg.~\ref{alg:RPE_noisy}) and MSQPE (Alg.~\ref{alg:mle-sinqpe}) in the presence of global depolarising noise.
        (Left) Numerical mean Holevo error $\epsilon$ as a function of the total cost $\mathcal{T}_{\mathrm{tot}}$ for a fixed noise rate $\gamma = 2^{-15}$.
        The vertical dashed lines represent $\mathcal{T}_{\mathrm{tot}}(\epsilon_1)$ and  $\mathcal{T}_{\mathrm{tot}}(\epsilon_1)$ separating different regimes in Alg.~\ref{alg:mle-sinqpe}.
        The solid lines correspond to the Cramer-Rao bound $\epsilon \geq \bar{\mathcal{I}}^{-1/2}$, where $\bar{\mathcal{I}}$ is the Fisher Information of all the quantum circuits used.
        (Center) Numerical error $\epsilon$ relative to the minimal error allowed by the Cramer-Rao bound $\bar{\mathcal{I}}^{-1/2}$.
        (Right) Performance in the limit $\epsilon \rightarrow 0$ as a function on the noise strength $\gamma$.
        The blue circles and orange pluses are obtained by fitting $\epsilon \propto \mathcal{T}_{\mathrm{tot}}^{1/2}$ to the numerical errors $\epsilon$ in the small error regime for different noise rates $\gamma$, and the solid lines by fitting $\bar{\mathcal{I}}^{-1/2} \propto \mathcal{T}_{\mathrm{tot}}^{1/2}$.
        The dotted lines represent the least-squares fit to these points.
        Error bars in all cases are insignificant.
    }
    \label{fig:sinqpe}
\end{figure}

\topic{Description of how figure was created} We compare numerically the MSQPE algorithm against RPE across a range of target precisions $\epsilon_\mathrm{t}$ and error rates $\gamma$.
For RPE, we choose the optimized parameters from Ref.~\cite{belliardoAchieving2020}, fixing the target precision $\epsilon_\mathrm{t}=2^{-j}$ to give optimal performance for the algorithm.
(As currently defined in Alg.~\ref{alg:RPE_noisy}, to target $2^{-(j-1)}> \epsilon_\mathrm{t}>2^{-j}$ has the same cost as estimating at a precision $2^{-j}$, though this could be improved upon - e.g. by extending Alg.~\ref{alg:parametrised_RPE} in a similar way as in Alg.~\ref{fig:sinqpecircuit_arb_dim}.)
To calculate the Holevo error $\epsilon_H^{(\mathrm{RPE})}$ achieved by RPE, we run $1000$ simulations of Alg.~\ref{alg:RPE_noisy} targeting a different random phase on $[0, 2\pi)$.
This yields the orange pluses in Fig.~\ref{fig:sinqpe}.
With the Holevo errors $\epsilon_H^{(\mathrm{RPE})}$ from RPE determined, we use this as a target error $\epsilon_\mathrm{t}=\epsilon_H^{(\mathrm{RPE})}$ for MSQPE. 
We run $1000$ simulations resulting in the blue circles in Fig.~\ref{fig:sinqpe}.
We observe that the mean Holevo error $\epsilon_H^{(\mathrm{MSQPE})}$ for MSQPE is very close to the target precision $\epsilon_\mathrm{t}$, which validates our choice of error model used to set the parameters of the algorithm (Eq.~\eqref{eq:mle_error_model}).
For comparison, we plot the Cramer-Rao bounds from the simulated quantum device data that is fed into either algorithm (orange for RPE and blue for SinQPE).
We additionally state the standard Heisenberg limit and the global depolarizing noise limit calculated in Ref.~\cite{kolodynski2013efficient} (see App.~\ref{sec:qpe_background}).
To construct Fig.~\ref{fig:MSQPE_vs_RPE}, we repeat the above calculation for different noise rate $\gamma = 2^{-m}$, and plot the relative total cost $\mathcal{T}_{\mathrm{tot}}$ of RPE vs MSQPE at a fixed $\epsilon_H^{(\mathrm{RPE})}\approx\epsilon_H^{(\mathrm{MSQPE})}$.

\topic{summary of figures}
We observe that the MSQPE achieves the target precision as expected; the orange and blue points on Fig.~\ref{fig:sinqpe}[left] have nearly-identical y-values, but differ in the cost $\mathcal{T}_{\mathrm{tot}}$.
The numerical error for the MSQPE is very close to the classical Cramer-Rao bound for the data taken, whilst the error of RPE is about $50\%$ above its Cramer-Rao bound (Fig.~\ref{fig:sinqpe} middle panel).
However, the Cramer-Rao bounds for the two algorithms are different (as they take different data), so this does not directly translate into a performance advantage of MSQPE at all regimes in Fig.~\ref{fig:MSQPE_vs_RPE}.
In the regime of large target error $\epsilon_\mathrm{t}$, the noise is not significant, and both algorithms perform as in the noiseless case, achieving near Heisenberg scaling $\epsilon_H \propto 1/\mathcal{T}_{\mathrm{tot}}$ (Fig.~\ref{fig:sinqpe}[left]).
However, MSQPE has a better scaling constant; it satisfies the Heisenberg limit $\epsilon_H \approx \pi/\mathcal{T}_{\mathrm{tot}}$, while RPE has a small gap $\epsilon_H \approx 5\pi/\mathcal{T}_{\mathrm{tot}}$.
This matches known results for the noiseless case~\cite{vandamOptimal2007,najafiOptimum2023}.
The advantage of MSQPE continues in the intermediate regime, and is more significant for smaller noise strengths $\gamma$ (as can be seen in Fig.~\ref{fig:MSQPE_vs_RPE}).
In the limit of small errors both algorithms asymptotically show $\epsilon_H \propto \sqrt{\gamma/\mathcal{T}_{\mathrm{tot}}}$ as expected.
In this regime, RPE outperforms the MSQPE by a factor of $\approx 1.3$.
Note that MSQPE continues to closely follow the Cramer-Rao bound in this case, which implies that the performance loss is driven by the sin-state being suboptimal in the presence of noise, rather than the classical estimation routine performing badly. For MSQPE, the scaling factor is $\approx 4.9$ (which is close to the analytical result in Thm.~\ref{thm:mle_sinqpe}, as $\sqrt{C}\approx 4.6$), and for RPE $\approx 4.0$ - slightly worse than $\sqrt{4e}\approx 3.3$ bound from the specific Fisher Information in Eq.~\eqref{eq:noisy rpe_lower_bound}.
Both are significantly ($2\times$) above the fundamental limit $0.89\sqrt{4\gamma/3e} \approx 1.7$ \cite{kolodynski2013efficient}.

\section{Error mitigation overhead for QPE}
\label{sec:mit_overhead_qpe}

The depolarizing noise model is a useful toy model to understand the idealized performance of a quantum algorithm (especially in terms of the variance cost), but it fails to capture the bias present in a typical noise model for a quantum device~\cite{caiQuantum2023}.
Over recent years, many error mitigation techniques have emerged, promising a range of fidelity overheads and unbiasing ability.
In this section we consider the effect realistic noise and error mitigation techniques would have on the results in App.~\ref{sec:qpe-global-depol}.
Our primary result is to incorporate a fidelity overhead of $F^{\alpha}$ directly into our $\mathcal{T}$ optimization for QPE.
We find that we can trade an increased $\alpha$ (typically $\alpha=1$, $2$, or $4$) against $\mathcal{T}$ to yield a multiplicative overall cost $\mathcal{T}_{\mathrm{tot}}\rightarrow\alpha\mathcal{T}_{\mathrm{tot}}$.
This suggests the overhead for explicitly unbiased methods such as probabilistic error cancellation~\cite{temmeError2017,endoPractical2018} may be affordable, but it also suggests that improving error mitigation will yield little further gains for QPE.

\label{app:error_mitigation_overhead}
\topic{general theory of EM} 
Error mitigation theory typically focuses on expectation value estimation tasks: measuring to precision $\epsilon$ an expectation value $\expval{O}{\psi}$ on a state $\ket{\psi}$ produced by a given noiseless quantum circuit \cite{caiQuantum2023}.
This can theoretically be obtained running the ideal circuit $M_1$ times.
Only having access to a noisy device, an error mitigation scheme attempts to reproduce the same estimate by running a larger number $M_F$ of shots of noisy quantum circuits.
The error mitigation overhead $C_\text{em}(F) = \frac{M_F}{M_1}$ depends on the fidelity $F=\bra{\psi}\rho\ket{\psi}$ to which the original circuit can be run on the noisy device (where $\ket\psi$ is the state that would be produced by a noiseless run of the logical circuit and $\rho$ is the mixed state actually produced by the noisy run on a device).

\begin{definition}[Error mitigation for expectation values]
    Consider a quantum circuit that prepares a state $\ket\psi$, a noisy device that can run this circuit prepating a mixed state $\rho$ with fidelity $F=\bra{\psi}\rho\ket{\psi}$, and an observable $O$ that can sampled on the prepared state.
    An \emph{error mitigation scheme} is an algorithm that produces an estimate $\widetilde{O}$ of $\bra{\psi} O \ket{\psi}$ by running $\widetilde{M}$ quantum circuits on the noisy device, and processing the outcome samples. (the circuits run can be different from the original logical circuit, and the algorithm can depend on additional knowledge of the device noise model.)
    An error mitigation scheme is \emph{unbiased} if it produces an estimate that converges to the true value $\widetilde{O} \to \bra{\psi} O \ket{\psi}$ in the large number of samples $\widetilde{M} \to \infty$.
    For an unbiased scheme, the \emph{error mitigation overhead} is defined as $C_{em}(F) = \lim_{\widetilde{M}\to\infty} \widetilde{M} \Var[\widetilde{O}]/\Var[O]$, where $\Var[\widetilde{O}]$ is the expected variance of the estimate and ${\Var[O] = \bra{\psi} O^2 \ket{\psi} -\bra{\psi} O \ket{\psi}^2}$.
\end{definition}

\topic{varieties of error mitigation techniques}
Error mitigation techniques can be classified based on the overhead scaling as a function of $F$.
If errors can be detected deterministically and associated runs discarded, error mitigation can be achieved with overhead $C_\text{em} \sim F^{-1}$. 
This is known as \emph{postselection bound}, and no error mitigation technique can perform better than this. 
An example of an error mitigation scheme whose overhead scales as $C_\text{em} \propto F^{-1}$ is symmetry verification \cite{bonet-monroigLowcostErrorMitigation2018,mcardleErrormitigatedDigitalQuantum2019}, and another is echo verification when state preparation is cheap compared to measurement \cite{obrienError2021}. 
Postselection schemes are typically problem dependent and can only mitigate part of the errors. 
A more general set of error mitigation schemes relies on measuring a noisy signal and processing it classically to remove the bias, which is estimated analytically under specific assumptions on the problem and noise model.
This is particularly effective for global depolarizing noise \cite{miInformation2021}, or when the effect of the noise can be simulated effectively \cite{filippovScalable2023}.
Due to error propagation in the required rescaling, the overhead for these techniques is $C_\text{em} \sim F^{-2}$: we call this \emph{rescaling bound}.
Note that echo verification and virtual distillation match this bound when the state preparation is expensive compared to measurement, as the volume of the noisy circuits that need to be run is double than the original circuit \cite{obrienError2021,hugginsVirtual2021}.
Finally, a last set of error mitigation techniques prescribes to measure explicitly the bias on the signal caused by circuit errors, and subtract it from the biased signal.
This is the case of probabilistic error cancellation (PEC), which can completely remove the bias from an observable expectation value by rewriting the noiseless signal as a non-convex combination of signals from noisy ``basis'' circuits.
In the case of local stochastic noise, the overhead of PEC is $C_\text{em} \propto F^{-4}$ \cite{temmeError2017,endoPractical2018,caiQuantum2023}.
This can be interpreted as a cost $F^{-2}$ due to the necessity of explicitly measuring the bias and $F^{-2}$ due to the necessity of rescaling of the resulting signal after subtracting the bias.

\topic{error mitigation for QPE}
When implementing quantum phase estimation in the presence of noise, we have an additional handle on the depth the circuits to be run; choosing this as a function of the noise rate impacts the scaling of the error mitigation overheads.
We consider a stochastic error model and we define the noise rate $\gamma$ as the logarithm of the error probability when implementing the noisy channel
\begin{equation}
\label{eq:noisy_unitary}
    \mathcal{U}(\rho) = e^{-\gamma} cU \rho cU^\dag + (1-e^{-\gamma}) \mathcal{N}(\rho),
\end{equation}
which approximates the controlled unitary oracle $cU$, where $\mathcal{N}$ is the action of the channel conditioned on at least one noise event happening (a.k.a.~pure noise channel).
In the noiseless case $\gamma=0$, we can achieve an accuracy $\leq \epsilon$ on the estimated phase $\phi$ by running the optimal SinQPE (Algorithm \ref{alg:sine-state-single-circuit}) at oracle depth $\mathcal{T}_{\text{tot}}=\lceil \pi/\epsilon \rceil$.
In the presence of noise, we can choose to run multiple circuits with smaller depth $\mathcal{T} < \mathcal{T}_\text{tot}$ resulting in fidelities bounded by $F > e^{-\gamma \mathcal{T}}$.
Under a given scheme, we define $\mathcal{T}_\text{tot}(\gamma, \epsilon)$ as the cost of estimating $\phi$ to accuracy $\epsilon$.
By analogy to the definition of $C_\text{em}$, we define the error mitigation overhead for this scheme as $C_\text{em}^\text{QPE}(\gamma, \epsilon) = \frac{\mathcal{T}_{\text{tot}}(\gamma, \epsilon)}{\mathcal{T}_{\text{tot}}(0, \epsilon)}$.

\begin{theorem}[Error mitigation overhead for QPE] \label{thm:em_overhead_qpe}
    Assume an unbiased error mitigation scheme for expectation values with overhead $C_\text{em}(F) = \Lambda F^{-\nu}$. Then we can realize an error-mitigated single-control QPE protocol that estimates $\phi$ to an accuracy $\epsilon$ with an overhead $C_\text{em}^\text{QPE}(\gamma, \epsilon) \propto \nu \gamma / \epsilon$ using the noisy implementation \(\mathcal{U}\) of $U$ in Eq.~\ref{eq:noisy_unitary}.
\end{theorem}
\begin{proof} 
    In the absence of noise, we know from Theorem \ref{thm:sinqpe-noiseless} that the optimal QPE algorithm (sin-state QPE - Alg.~\ref{alg:sine-state-single-circuit}) estimates $\phi$ with a precision $\epsilon$ in time $\mathcal{T}_\text{tot}^\text{opt} = \frac{\pi}{\epsilon}$.
    From Theorem \ref{thm:rpe_noiseless} we know that, using RPE, we can estimate $\phi$ to precision $\widetilde{\epsilon} \leq \epsilon$ using noiseless circuits of depth bounded by $\mathcal{T}_\text{max} = \frac{b}{\widetilde{\epsilon}}$ for a cost $\mathcal{T}_\text{tot}^\text{RPE}(\widetilde{\epsilon}) = \frac{a}{\widetilde{\epsilon}}$ (for algorithm-specific constants $a=5$ and $b$ analytically bounded by $b<24.26\pi$ and numerically found to be $b\approx5\pi$).
    We can recover an estimate to precision $\epsilon$ by repeating $\frac{\widetilde{\epsilon}^2}{\epsilon^2}$ times this RPE algorithm to precision $\widetilde{\epsilon}$ and averaging the results; the maximum depth for the RPE circuits remains unchanged, and the total time becomes $\mathcal{T}_\text{tot}(\mathcal{T}_\text{max}, \epsilon) = \frac{a \widetilde{\epsilon}}{\epsilon^2} = \frac{\kappa}{\mathcal{T}_\text{max} \epsilon^2}$ (where we defined the overall constant $\kappa:=ab$).
    As the single-control QPE circuits measure expectation values, we can apply error mitigation to recover the noiseless signal.
    The fidelity of these circuits is bounded by $F \geq e^{-\gamma \mathcal{T}}$, so the error mitigation overhead is bounded by $\mathcal{C}_\text{em} \leq \Lambda e^{\nu \gamma \mathcal{T}}$.
    The QPE error mitigation overhead is then calculated as 
    \begin{equation}
        \mathcal{C}_\text{em}^\text{QPE}(\gamma, \epsilon) = \frac{\mathcal{T}_{\text{tot}}^\text{sc} \mathcal{C}_{\text{em}}}{\mathcal{T}_{\text{tot}}^{\text{opt}}} \leq \frac{\kappa \Lambda}{\mathcal{T} \pi \epsilon} e^{\nu \gamma \mathcal{T}}
        \quad
        \xrightarrow[\text{over } \mathcal{T}]{\text{minimize}}
        \quad
        C_\text{em}^\text{QPE}(\gamma, \epsilon) 
        \leq
        \frac{\kappa \Lambda e}{\pi} \frac{\nu \gamma}{\epsilon},
    \end{equation}
    with the bound minimized by the choice of $\mathcal{T} = \frac{1}{\gamma \nu}$.
\end{proof}

\section{Explicit unbiasing of multi-control QPE}
\label{sec:unbiasing_qpe}

In sin-state QPE, the estimation task cannot be reduced to sample averaging, as the estimate of the phase is not an expectation value.
In this section, we extend explicit unbiasing techniques to a more complex data processing scheme: maximum likelihood estimation.
Extending error mitigation techniques beyond expectation value estimation is identified as a relevant open problem in literature \cite{caiQuantum2023}.

\subsection{The explicitly-unbiased maximum-likelihood estimator}
\label{sec:EUMLE}

\topic{PEC channel decomposition}
Similar to what is done in PEC for expectation values, we begin by decomposing the unitary channel that implements the ideal SinQPE circuit $\mathcal{U}[\rho] = U \rho U^\dag$ as a non-convex linear combination of ``basis operations'' $\{\mathcal{B}_j\}$, i.e.~channels that can be run on the noisy device:
\begin{equation} \label{eq:pec-decomposition}
    \mathcal{U} = \sum_j \alpha_j \mathcal{B}_j.
\end{equation}
For stochastic Pauli channels (which can represent a model for logical noise in surface code even when the underlying physical noise is coherent \cite{bravyiCorrecting2018}), the decomposition procedure is known.
In particular, for a circuit of volume $A$ and local depolarizing noise with constant rate per volume element $r$, the decomposition has 1-norm $|\boldsymbol{\alpha}|_1 = e^{2rA} = F^{-2}$ where $F = e^{-rA}$ is the circuit fidelity.
(The cost of standard PEC for mitigating expectation values is then $C_\text{em} = |\boldsymbol{\alpha}|_1^2 = F^{-4}$.)

\topic{Probabilities notation and decomposition}
In multi-circuit sin-state QPE, we are interested in estimating the phase $\phi$ given a set of $M$ measurements $\{x_i\}_{i=1,...,M}$.
In the ideal, noiseless case, these are i.i.d.~samples obtained from the SSQPE circuit (see Fig.~\ref{fig:sinqpecircuit}) with distribution $x \sim P(x) = \Tr{\Pi_x \mathcal{U}[\rho_0]}$, where $\{\Pi_x\}$ is the projector of the control register onto the bit-string $x$ and $\rho_0$ is the input state.
We can decompose this distribution using Eq.~\eqref{eq:pec-decomposition} as
\begin{equation}
\label{eq:pec-probability-decomposition}
    P(x) = \sum_j \alpha_j Q_j(x)
    \quad,\quad
    Q_j(x) = \Tr{\Pi_x \mathcal{B}_j[\rho_0]}.
\end{equation}
Our goal is to define a consistent estimator of $\phi$ based on samples taken from the latter distributions, from which we can sample using the noisy device.

\topic{Maximum likelihood model and decomposition}
In the case of multiple samples $\{x_i\}$ taken from noiseless circuits, an optimal estimator $\hat\phi$ of $\phi$ is defined by maximising the log-likelihood of the samples over the parametric model of the outcome distribution $P(x|\phi)$ given in Eq.~\ref{eq:sinqpe_noisless_prob_distribution}:
\begin{equation}
    \hat{\phi} = \argmax_\phi \mathcal{L}(\phi)
    \quad , \quad 
    \mathcal{L}(\phi) = \frac{1}{M} \sum_{i=1}^M \log P(x_i|\phi).
\end{equation}
We added a constant factor $M^{-1}$ compared to the usual definition of log-likelihood, making clear that the $\mathcal{L}(\phi)$ is evaluated as a sample average of the logarithm of the model probability.
In the large sample limit,
\begin{equation} \label{eq:loglikelihood-limit}
   \lim_{M\to\infty} \mathcal{L}(\phi) = 
   \bar{\mathcal{L}}(\phi) = 
   \E\!\big[\log P(x|\phi) \big| x \sim P(x)\big],
\end{equation}
which we can decompose through Eq.~\eqref{eq:pec-probability-decomposition} as
\begin{align}
    \bar{\mathcal{L}}(\phi) = & 
    \sum_j \alpha_j \E\!\big[\log P(x|\phi) \big| x \sim Q_j(x)\big]
    \\
    = & 
    |\boldsymbol{\alpha}|_1
    \E\!\left[\sgn(\alpha_j)
        \E[\log P(x|\phi) | x \sim Q_j(x)]
        \Big| j \sim \frac{|\alpha_j|}{|\boldsymbol{\alpha}|_1}
    \right]
    \\
    = &
    |\boldsymbol{\alpha}|_1
    \E\!\left[
        \sgn(\alpha_j) \log P(x|\phi) 
        \Big| j, x \sim \widetilde{Q}(j, x)
    \right]\label{eq:likelihood_as_rescaled_distributions}
\end{align}
where we defined the joint distribution $\widetilde{Q}(j, x) = \frac{|\alpha_j|}{|\boldsymbol{\alpha}|_1} \cdot Q_j(x)$. 
We can sample data points $\{(j_i, y_i)\}_{i=1,...,\widetilde{M}}$ by importance-sampling $j_i \sim |\alpha_{j_i}|/|\boldsymbol{\alpha}|_1$ and then sampling $y_i \sim Q_{j_i}(y_i)$ on the noisy quantum device.
We can then compute an estimate $\widetilde{\mathcal{L}}$ of $\bar{\mathcal{L}}$ as a sample average
\begin{equation}
\label{eq:likelihood-no-regularisation}
    \widetilde{\mathcal{L}}(\phi) = \frac{|\boldsymbol\alpha|_1}{\widetilde{M}} \sum_{i=1}^{\widetilde{M}} \sgn(\alpha_{j_i}) \log P(y_i|\phi).
\end{equation}
Maximising this gives an estimator $\widetilde{\phi}$ of $\phi$, which we define below.
We call this estimator the \emph{explicitly unbiased maximum likelihood estimator} (EUMLE).

\begin{definition}\label{def:eumle}
    [Explicitly unbiased maximum likelihood estimator]
    Given (1) a parameterized model $P(x | \phi)$ for the distribution of $x$ obtained from a noiseless quantum circuit, (2) a non-convex decomposition of $P(x | \phi^*) = \sum_j \alpha_j Q_j(x)$ in terms of distributions $Q_j(x)$ that can be sampled using a noisy quantum device; the explicitly unbiased maximum likelihood estimator (EUMLE) of the phase $\phi^*$ is defined as $\widetilde\phi = \argmax_\phi \widetilde{\mathcal{L}}(\phi)$, where $\widetilde{\mathcal{L}}(\phi)$ in Eq.~\eqref{eq:likelihood-no-regularisation} is constructed by sampling $\widetilde{M}$ points $\{(j_i, x_i) | j_i\sim|\alpha_{j_i}|/|\boldsymbol{\alpha}|_1, x_i \sim Q_{j_i}(x_i)\}_{i=1,...,M}$.
\end{definition}
The EUMLE is a consistent estimator, i.e.~in large samples limit $\widetilde{M} \to \infty$ it converges to the true value of $\phi$.
In the absence of regularization ($c=0$), $\widetilde{\mathcal{L}}_c(\phi) = \widetilde{\mathcal{L}}(\phi) \to \bar{\mathcal{L}}(\phi)$ asymptotically matches the log-likelihood of the noiseless data Eq.~\eqref{eq:loglikelihood-limit}.
Here we prove the consistency of the EUMLE by showing that, even for $c>0$, the maximum of the likelihood for $\widetilde{M} \to \infty$ remains the same. 
We also prove the stronger property of asymptotic normality and recover an expression for the asymptotic variance of the estimator.

\begin{theorem}[asymptotic normality of the EUMLE]\label{thm:eumle}
    The EUMLE $\widetilde{\phi}$ of the phase $\phi$ is asymptotically normal with mean $\phi$ and variance $\sigma^2 \propto {\widetilde{M}}^{-1}$.
\end{theorem}
\begin{proof}
    For this proof we follow a strategy similar to the one use to prove asymptotic normality of the maximum likelihood estimator in Theorem 10.1.12 of \cite{casellaStatistical2002}.
    First, we Taylor-expand $\tilde{\mathcal{L}}'(\phi)$ (the primes are derivatives w.r.t.~$\phi$) around the true value $\phi^* = \argmin_\phi{\bar{\mathcal{L}}(\phi)}$:
    \begin{equation}
        \tilde{\mathcal{L}}'(\phi) =
        \tilde{\mathcal{L}}'(\phi^*) + 
        (\phi - \phi^*) \tilde{\mathcal{L}}''(\phi^*) + 
        \mathcal{O}((\phi - \phi^*)^2).
    \end{equation}
    We ignore the second order term as it is subdominant for small errors under reasonable regularity conditions, and we use the stationary point condition $\tilde{\mathcal{L}}'(\tilde\phi) = 0$ to get
    \begin{equation}
        \tilde\phi - \phi^* =  \frac{- \tilde{\mathcal{L}}'(\phi^*)}{\tilde{\mathcal{L}}''(\phi^*)}.
    \end{equation}
    The denominator is
    \begin{align}
        \tilde{\mathcal{L}}''(\phi^*) &=
        \frac{|\boldsymbol\alpha|_1}{\tilde{M}}
        \sum_{i=1}^{\tilde{M}} \sgn(\alpha_{j_i}) \, 
        [\partial_\phi^2 \log P(y_i|\phi)]_{\phi^*}
        \\
        &\xrightarrow[\tilde{M} \to \infty]{\text{p}}
        |\boldsymbol\alpha|_1 
        \E\!\big[\sgn(\alpha_j) \, [\partial_\phi^2 \log P(y|\phi)]_{\phi^*} \big| j, y \sim \tilde{Q}(j, y)\big]
        \\ & =
        \sum_j \alpha_j \int \dd Q_j(y) \,\,
        [\partial_\phi^2 \log P(y|\phi)]_{\phi^*}
        \\ & =
        \E\!\big[[\partial_\phi^2 \log P(y|\phi)]_{\phi^*} \big| y \sim P(y)\big]
        =
        \E\!\big[\underbrace{[\partial_\phi \log P(y|\phi)]^2_{\phi^*}}_{[s(y)]^2} \big| y \sim P(y)\big]
        =: \mathcal{I}(\phi^*)
    \end{align}
    where the convergence in probability $\xrightarrow[]{\text{p}}$ is justified by the law of large numbers, and $\mathcal{I}(\theta^*)$ is the Fisher information of an ideal sample $x$ from the noiseless distribution $P(x)$.
    For brevity of further notation, we define the score of the probability model $s(x) := [\partial_\phi \log P(x|\phi)]_{\phi=\phi^*}$.
    By the central limit theorem, the numerator
    \begin{align}
        -\tilde{\mathcal{L}}'(\phi^*) &= 
        -\frac{|\boldsymbol\alpha|_1}{\tilde{M}}
        \sum_{i=1}^{M'} \sgn(\alpha_{j_i}) s(y_i)
    \end{align}
    converges in distribution to a normal distribution with mean $-\E[s(x) | x \in P(x)] = 0$ and variance 
    $\frac{|\boldsymbol\alpha|_1^2}{\tilde{M}}\E\!\big[[s(y)]^2 \big| y \sim \tilde{Q}(y)\big]$,
    with $\tilde{Q}(y) = \sum_j \tilde{Q}(j, y)$ notating the marginal distribution.
    Thus, the estimator error $\tilde{\phi} - \phi$ is asymptotically normal with mean $0$ and variance 
    \begin{equation}
    \label{eq:eumle_variance}
        \sigma^2 = \frac{|\boldsymbol\alpha|_1^2}{\tilde{M}}\frac{\E\!\big[s^2(y) \big| y \sim \tilde{Q}(y)\big]}{\E\!\big[s^2(y) \big| y \sim P(y)\big]^2}.
    \end{equation}
\end{proof}

\subsection{Regularizing the EUMLE}
\label{sec:rEUMLE}

Maximizing $\widetilde{\mathcal{L}}$ as-is leads to unstable behaviour as our probability distribution $P(x|\phi)$ [Eq.~\eqref{eq:sinqpe_noisless_prob_distribution}] can be arbitrarily close to $0$.
This is especially problematic in our case, as we are sampling from $\widetilde{Q}(x)$ instead of $P(x|\phi)$, and so we do not expect singular behaviour of $\widetilde{\mathcal{L}}$ to be a rare event.
To alleviate this problem we regularize our model $P(x|\phi)$ by adding a small constant; $P(x|\phi)\rightarrow P_c(x|\phi)=P(x|\phi)+c$.
(In principle we can add an arbitrary $\phi$-independent function to $P$, which may be better for specific noise cases. This extension is left for future research.)
This yields a new form of the log likelihood
\begin{equation} \label{eq:regularized-likelihood-from-noisy-data}
    \widetilde{\mathcal{L}}_c(\phi) = 
    \frac{|\boldsymbol\alpha|_1}{\widetilde{M}} \sum_{i=1}^{\widetilde{M}} \sgn(\alpha_{j_i}) \log P_c(y_i|\phi) + \frac{c}{|\mathcal{D}|} \sum_{x \in \mathcal{D}} \log P_c(x|\phi),
\end{equation}
where $\mathcal{D}$ is the domain of the samples $x$ (e.g. strings of $\log_2(\mathcal{T})$ bits).

For periodic distributions with continuous domain $\mathcal{D}=[0, 2\pi)$ where $P(x|\phi) = P(x-\phi \mod2\pi)$ holds, the second term is not dependent on $\phi$ and can be neglected.
Such a distribution can be obtained from QPE adding a random offset to the phase being estimated \cite{lindenAverageCase2022,apeldoornQuantum2022}.
Here, however, we introduce a general regularized EUMLE which uses this modified likelihood and prove its asymptotic normality.

\begin{definition}\label{def:reumle}
    [regularized EUMLE]
    The regularized Explicitly Unbiased Maximum Likelihood Estimator (rEUMLE) of the phase $\phi^*$ is derived from the EUMLE (definition \ref{def:eumle}); it requires the same inputs plus an additional regularization constant $c$, it uses the same samples from the quantum device, and is defined as $\tilde{\phi} = \argmax_\phi{\widetilde{\mathcal{L}}_c}$ using the regularized log likelihood Eq.~\ref{eq:regularized-likelihood-from-noisy-data}.
\end{definition}

\begin{theorem}[asymptotic normality of the rEUMLE]\label{thm:eumle_regularised}
    The rEUMLE $\widetilde{\phi}$ of the phase $\phi$ is asymptotically normal with mean $\phi^*$ and variance \begin{equation} \label{eq:reumle-variance}
        \sigma^2 = \frac{|\boldsymbol\alpha|_1^2}{\widetilde{M}}\frac{\E\!\big[s_c^2(y) \big| y \sim \widetilde{Q}(y)\big]}{\E\!\big[s_c^2(y) \big| y \sim P(y)\big]^2},
    \end{equation}
    where $s_c$ is the regularised score function $s(x) := [\partial_\phi \log P_c(x|\phi)]_{\phi=\phi^*}$.
\end{theorem}
\begin{proof}

 First, we show that the estimator is consistent, i.e. in the limit of infinite samples $\widetilde{M}$ the estimator converges to the true value $\phi^*$. The regularised likelihood function $\widetilde{\mathcal{L}}_c$ converges to the expectation value
    \begin{align}
        \mathcal{L}_c(\phi) \equiv \lim_{\widetilde{M}\to\infty} \widetilde{\mathcal{L}}_c(\phi) &= |\boldsymbol\alpha|_1 \E \bigg[\sgn(\alpha_{j}) \log P_c(x|\phi)\ \big| \ (j, x) \sim \frac{|\alpha_{j}|}{|\boldsymbol{\alpha}|_1} Q_{j}(x)\bigg]
        +
        \frac{c}{|\mathcal{D}|} \sum_{x \in \mathcal{D}} \log P_c(x|\phi) 
        \\
        & = \frac{1}{|\mathcal{D}|} \sum_{x \in \mathcal{D}} P(x)\log P_c(x|\phi) + c \log P_c(x|\phi)
        \\
        &= \frac{1}{|\mathcal{D}|} \sum_{x \in \mathcal{D}} P_c(x)\log P_c(x|\phi).
    \end{align}
    $\mathcal{L}_c(\phi)$ has a maximum at $\phi = \phi^*$ since
    \begin{align}
        \mathcal{L}_c'(\phi^*) &= \partial_\phi  \left[\frac{1}{|\mathcal{D}|} \sum_{x \in \mathcal{D}} P_c(x|\phi)\right] = \partial_\phi[1 + c] = 0\\
        \mathcal{L}_c''(\phi^*) &=  -
        \frac{1}{|\mathcal{D}|} \sum_{x \in \mathcal{D}}
        \frac{[\partial_\phi P(x|\phi)]^2_{\phi = \phi^*}}{P_c(x|\phi^*)} < 0.
    \end{align}
    The proof of asymptotic normality then follows the same steps as the proof of Thm.~\ref{thm:eumle} using the regularized versions of the likelihood $\widetilde{\mathcal{L}}_c$ and score $s_c$.
\end{proof}

To conclude, we calculate the sample overhead of this error mitigation technique.
This is defined as the ratio between the number of samples $\widetilde M$ needed to achieve $\sigma^2$ to the number of samples $M$ that would be sufficient in the noiseless case ($\widetilde{Q} = P$, $\alpha$ = 1 and no regularization required $c=0$):
\begin{equation}\label{eq:eumle_overhead}
    \mathcal{C}_\text{EUMLE} =  \frac{\widetilde M}{M} = 
    |\boldsymbol\alpha|_1^2 \frac{\E\big[s_c^2(y)\big| y \sim \widetilde{Q}(y)\big]}{\E\big[s_c^2(x) \big| x \sim P(x)\big]}
    \,
    \frac{\E\big[s_0^2(x)\big| x \sim P(x)\big]}{\E\big[s_c^2(x) \big| x \sim P(x)\big]}.
\end{equation}
Here, the term $\mathbb{E}[s_0^2(x)|x\sim P(x)]$ is the Fisher information of the noiseless distribution.
The overhead consists of three factors: 1.~the rescaling constant $|\boldsymbol\alpha|_1^2$, 2.~the ratio between the average (regularized) score of and the importance-sampled noisy distribution and the noiseless distribution, and 3.~the overhead due to the loss of information caused by regularization.
The regularization constant $c$ ensures term 2 remains bounded, but it also causes term 3 to increase, so it should be chosen to optimize the product of the two.
Quantifying $\mathcal{C}_\text{EUMLE}$ requires additional assumptions about the noise model and the consequent basis circuit decomposition, and about the specific model $P(x|\phi)$.

\subsection{Applying the regularized EUMLE to sin-state QPE}
\label{sec:EUMLE-sinQPE}

In this section, we explore the application of the rEUMLE to phase estimation, using the samples from the SinQPE circuit.
The regularization is necessary as the SinQPE distribution $P(x|\phi)$ vanishes for certain values of $\phi$.
Here we bound the variance of the resulting rEUMLE.

\begin{theorem}
\label{thm:eumle_qpe}
    [EUMLE for Sin QPE] We consider the rEUMLE (Def.~\ref{def:reumle}) applied to the SinQPE probability distribution [Eq.~\eqref{eq:sinqpe_noisless_prob_distribution}], which can be expressed with a PEC decomposition [Eq.~\eqref{eq:pec-probability-decomposition}] with one-norm $|\boldsymbol\alpha|_1$.
    For any constant $c$, there exists a threshold depth $\mathcal{T}_\mathrm{thr}$ and a constant $B$ such that, for any SinQPE circuit with depth $\mathcal{T} > \mathcal{T}_\mathrm{thr}$, the variance of the EUMLE $\sigma^2$ is bounded as
    \begin{equation}\label{eq:eumle_qpe_var}
        \sigma^2 \leq B\frac{|\boldsymbol\alpha|_1^2}{M\mathcal{T}^2},
    \end{equation}
    where $M$ is the number of samples and the total oracular cost of the algorithm is $\mathcal{T}_\mathrm{tot} = M \mathcal{T}$.
\end{theorem}
\begin{proof}

The goal of this theorem is to bound the variance of the rEUMLE Eq.~\ref{eq:reumle-variance} for the sin-state QPE outcome probability distribution (from Theorem \ref{thm:mle_sinqpe})
\begin{align}
    P_{\mathcal{T}}(x|\phi) &= f_{\mathcal{T}}\left(\phi - \frac{2\pi x}{\mathcal{T}+1}\right), \\
    f_{\mathcal{T}}(\Delta\phi) &= \frac{\sin^2\frac{\pi}{\mathcal{T}+2}}{(\mathcal{T}+1)(\mathcal{T}+2)} \frac{1+\cos[(\mathcal{T}+2)\Delta\phi]}{\left[\cos(\Delta\phi) - \cos(\frac{\pi}{\mathcal{T}+2})\right]^2}.
\end{align}
In order to bound $\sigma^2$, we calculate asymptotic bounds for $\E\!\big[s_c^2(y) \big| y \sim P(y)\big]$ (bounded from below as $\Omega(\mathcal{T})^2$) and for
$\E\!\big[s_c^2(y) \big| y \sim \widetilde{Q}(y)\big]$ (bounded from above as $\mathcal{O}(\mathcal{T}^2)$).
This allows us to bound $\sigma^2$ for sufficiently large $\mathcal{T}$.

We first show that $\frac{1}{\*{T}^2} \E\!\big[s_c^2(y) \big| y \sim P(y|\phi)\big]$ converges to a constant $a$ in the $\mathcal{T}\rightarrow \infty$ limit.
Substituting the definition of the probability we can write
\begin{align}
    \frac{1}{\mathcal{T}^2}\E\!\big[s_{c}^2(y) \big| y \sim P(y | \phi)\big]
    &= 
    \sum_{x=0}^{\mathcal{T}}\frac{\mathcal{T}^{-2} f'_\mathcal{T}\left(\phi - \frac{2\pi x}{\mathcal{T}+1}\right)^2}{\left[f_\mathcal{T}\left(\phi - \frac{2\pi x}{\mathcal{T}+1}\right) + c \right]^2}  f_\mathcal{T}\left(\phi - \frac{2\pi x}{\mathcal{T}+1}\right).
\end{align}
For $\mathcal{T} \to \infty$, we can replace the normalized sum $\frac{1}{\mathcal{T}+1} \sum_{x=0}^\mathcal{T}$ with an integral
\begin{align}
    \frac{1}{\mathcal{T}^2}\E\!\big[s_{c}^2(y) \big| y \sim P(y | \phi)\big]
    \to& 
    \frac{\mathcal{T}+1}{2\pi}
    \int_0^{2\pi} ds \,
    \frac{\mathcal{T}^{-2} f'_\mathcal{T}\left(\phi - s \right)^2}{\left[f_\mathcal{T}\left(\phi - s\right) + c \right]^2}  f_\mathcal{T}\left(\phi - s\right)
    \\
    =&
    \frac{\mathcal{T}+1}{2\pi}\int_{-\pi}^{\pi} ds \,
    \frac{\mathcal{T}^{-2} f'_\mathcal{T}\left(s \right)^2}{\left[f_\mathcal{T}\left(s\right) + c \right]^2}  f_\mathcal{T}\left(s\right)
    \\
    =&
    \frac{\mathcal{T}+1}{2\pi(\mathcal{T}+2)}\int_{-(\mathcal{T}+2) \pi}^{(\mathcal{T}+2) \pi} ds' \,
    \frac{\mathcal{T}^{-2} f'_\mathcal{T}\left(\frac{s'}{\mathcal{T}+2}\right)^2}{\left[f_\mathcal{T}\left(\frac{s'}{\mathcal{T}+2}\right) + c \right]^2}  f_\mathcal{T}\left(\frac{s'}{\mathcal{T}+2}\right)
\end{align}
where, in the second step, we used the periodicity and evenness of $f_\mathcal{T}$ to remove the dependence on $\phi$ and shift the integration bounds, and in the third step we performed a change of variable $s' = (\mathcal{T}+2) s$.
The $\mathcal{T}\to\infty$ limit of this integral can be calculated following the same steps as in the proof of Thm.~\ref{thm:mle_sinqpe}.
Using the convergence property $\lim_{\mathcal{T}\to\infty} f(s/(\mathcal{T}+2)) = f(s)$ with $f(s)$ in Eq.~\eqref{eq:limit_prob_func}, we can write
\begin{align}
    \frac{1}{\mathcal{T}^2}\E\!\big[s_{c}^2(y) \big| y \sim P(y | \phi)\big]
    \xrightarrow{\mathcal{T}\to\infty}
    \int_{-\infty}^{\infty} \frac{ds}{2\pi}
    \frac{f'\left(s\right)^2}{\left[f\left(s\right) + c \right]^2}  f\left(s\right) =: a.
\end{align}
This integral converges to a finite constant $a$, as the integrand (independent on $\mathcal{T}$) can be bounded by $j(s)$ defined in Eq.~\eqref{eq:limit-score}, which can be bounded by $\mathcal{O}(|s|^{-4})$ [see Eq.~\eqref{eq:limit-score-bound}].

Secondly, we show $\lim_{\mathcal{T}\rightarrow \infty}\frac{1}{\mathcal{T}^2}\E\!\big[s_c^2(y) \big| y \sim \widetilde{Q}(y)\big]$ is bounded by bounding $\lim_{\mathcal{T}\rightarrow \infty}\frac{1}{\mathcal{T}}|s_c(x)|$ for any $x$.
By definition, the score is,
\begin{equation}
    s_c(x) = [\partial_\phi \log P_c(x|\phi)]_{\phi=\phi^*} = \frac{f_{\mathcal{T}}'(\frac{2\pi x}{\mathcal{T}+1}-\phi)}{f(\frac{2\pi x}{\mathcal{T}+1}-\phi) + c},
\end{equation}
and we can bound it as $|s_c(x)|\leq |f_{\mathcal{T}}'(\frac{2\pi x}{\mathcal{T}+1}-\phi)|/c$.
As per Eq.~\eqref{eq:limit_prob_deriv}, the limit function $f'(s) = \lim_{\mathcal{T}\rightarrow\infty}\frac{1}{\mathcal{T}}f'(s/(\mathcal{T}+2))$ is continuous and $f'(s)\xrightarrow{s\rightarrow\pm\infty}0$ so $f'$ is bounded by a constant $b := \sup_s |f'(s)|$.
Combining all the steps, we get the bound
\begin{equation}
    \lim_{\mathcal{T}\rightarrow\infty}\frac{1}{\mathcal{T}^2} \E\!\big[s_c^2(y) \big| y \sim \widetilde{Q}(y)\big]
    \leq \lim_{\mathcal{T}\rightarrow\infty}\frac{1}{\mathcal{T}^2}\max_x|s_c^2(x)|
    \leq \lim_{\mathcal{T}\rightarrow\infty}\frac{1}{\mathcal{T}^2}\frac{\max_s|f'_{\mathcal{T}}(s)|^2}{c^2}\leq \frac{b^2}{c^2}.
\end{equation}

Combining the above limits,
\begin{equation}
    \lim_{\mathcal{T}\to\infty}\mathcal{T}^2\frac{\E\!\big[s_c^2(y) \big| y \sim \widetilde{Q}(y)\big]}{\E\!\big[s_c^2(y) \big| y \sim P(y)\big]^2} 
    \leq
    \frac{b^2}{a^2 c^2}.
\end{equation}
For any constant $B > {b^2}/{a^2 c^2}$ we can then use this limit to bound the variance $\sigma^2$ for any sufficienly large depth $\mathcal{T}>\mathcal{T}_B$:
\begin{equation}
    \sigma^2 
    = 
    \frac{\E\!\big[s_c^2(y) \big| y \sim \widetilde{Q}(y)\big]}{\E\!\big[s_c^2(y) \big| y \sim P(y)\big]^2} 
    \frac{|\boldsymbol\alpha|_1^2}{\widetilde{M} \mathcal{T}^2}
    \leq
    B
    \frac{|\boldsymbol\alpha|_1^2}{\widetilde{M} \mathcal{T}^2}.
\end{equation}

Numerically, we find that the limits for $\mathcal{T}\to \infty$ act as bounds in the correct direction: for any $\mathcal{T}$, $\max f'_\mathcal{T} \leq (\mathcal{T}+2)\max f'$, and $\E\!\big[s_c^2(y) \big| y \sim P(y)\big] \geq a (\mathcal{T}+2)^2$.
We find the value of $b\approx 0.16$.
For an example choice of $c=0.4$, we find $a\approx 0.01$.
Hence, the bound on the variance Eq.~\eqref{eq:eumle_qpe_var} is satisfied for any depth $\mathcal{T}$ for $B=b^2/a^2c^2\approx12.6$.
\end{proof}

\subsection{Overhead for QPE with local Pauli noise}
\label{sec:EUMLE_overhead}

\topic{bound including fidelity}
To quantify the EUMLE overhead in quantities that we can relate to the rest of this work, we assume a local stochastic noise channel with uniform error rate that is locally invertible \cite{caiQuantum2023,temmeError2017,endoPractical2018}.
(This is the case for local Pauli channels with constant and qubit-independent error rate, e.g.~local depolarizing noise.)
Under this assumption, Eq.~\eqref{eq:pec-decomposition} can be written
\begin{equation}
    \mathcal{U} = \frac{1}{F}\mathcal{B}_0 + \sum_{j>0} \alpha_j\mathcal{B}_j
    \,,\quad
    |\boldsymbol{\alpha}|_1 = \frac{1}{F^2}
\end{equation}
where $F$ is the logical circuit fidelity and $\alpha_0 = 1/F$ is the weight of the noisy implementation $\mathcal{B}_0$ of the logical circuit.

It remains to optimize $\mathcal{T}$, in an equivalent manner to Thm.~\ref{thm:em_overhead_qpe}.
To achieve precision $\epsilon$ using a circuit of depth $\mathcal{T}$, we need to choose the number of samples $\widetilde{M} = \sigma^2/\epsilon^2$.
Using Theorem~\ref{thm:eumle_qpe}, we can write
\begin{equation}
\mathcal{T}_{\mathrm{tot}} = \mathcal{T}\widetilde{M}=\mathcal{T}\frac{\sigma^2}{\epsilon^2} \leq \frac{B|\boldsymbol\alpha|_1^2}{\epsilon^2\mathcal{T}}.
\end{equation}
Using $|\boldsymbol\alpha|_1 = F^{-2}$ and $F = e^{-\gamma \mathcal{T}}$
\begin{equation}
    \mathcal{T}_{\mathrm{tot}} \leq \frac{Be^{-4\gamma \mathcal{T}}}{\epsilon^2\mathcal{T}}
\end{equation}
This is optimized at $\mathcal{T}=\frac{1}{4\gamma}$, yielding
\begin{equation}
    \mathcal{T}_{\mathrm{tot}}^{(\mathrm{Pauli\, noise})} = 4B e\gamma\epsilon^{-2} \approx 137 \gamma\epsilon^{-2}.
\end{equation}
Comparing to Thm.~\ref{thm:eumle}, where $\mathcal{T}_{\text{tot}} = C \gamma \epsilon^{-2}$ with $C\approx 20$, we find the overhead from mitigating local Pauli noise to be a constant factor $7\times$ larger.
As the bound in Thm.~\ref{thm:eumle_qpe} holds for arbitrary adversarial noise, while the expression in Thm.~\ref{thm:eumle} saturated in the limit $\epsilon \to 0$, we expect in practice this gap to be significantly smaller.
Furthermore, one could consider replacing our constant regularization with a noise-specific additional term, which may lead to further optimization.

\subsection{Improving the overhead by filtering}\label{sec:filtering}

In order to improve the error mitigation overhead of the EUMLE, we propose to combine the explicit unbiasing method with a postselection method, which we refer to as filtering.
In the single eigenstate case, we expect the noiseless samples from the QPE circuit to cluster in some region, around the true phase value.
By postselecting the outcomes of all our circuit runs based on whether they lie in this clustering region, we can avoid the explosion of the final estimator variance due to the spread of the noisy samples.
Note that this is already somewhat achieved by the regularization, as the regularized likelihood is constant on phase estimates sufficiently far from the true value.
However, filtering gives us a means to further add prior information, such as energy estimates from classical methods, to boost our QPE performance.

We develop here a filtering procedure that assumes prior knowledge of a clustering region of size $\Delta \gg \epsilon$ where the phase should lie.
The proposed filtering procedure can be applied offline, as the circuits that need to be sampled are the same described above for the EUMLE, the only difference being in the data processing.

\topic{Filtering procedure}
Given a guess $\phi_\text{guess}$ for the phase we are attempting to estimate and a phase distance $\Delta$ within which we expect to find the true phase, we define the filtering region $[\phi_\text{guess} - \Delta, \phi_\text{guess} + \Delta]$ and the relative indicator function 
\begin{equation}
    R(x) = [1 \text{ if } |x - \phi_\text{guess}| < \Delta, \, \text{else } 0].
\end{equation}
We start from a set of measurements $\{(j_i, y_i)\}_{i=1,...,\widetilde{M}}$ sampled from the same distribution $\widetilde{Q}(j, y)$ defined in the section above, obtained by importance-sampling $j$ from the PEC decomposition Eq.~\eqref{eq:pec-probability-decomposition}.
We then reject all samples that lie outside the filtering region, obtaining a subset of the data 
\begin{equation} \label{eq:filtered-subset}
    \mathcal{S} = \Big\{(j_i, y_i) \in \{(j_i, y_i)\}_{i=1,...,\widetilde{M}} \Big|  y_i \in [\phi_\text{guess} - \Delta, \phi_\text{guess} + \Delta] \Big\}.
\end{equation}
This subset represents samples from the filtered distribution $\widetilde{Q}(j, y) \cdot R(y)$
(up to an irrelevant normalization factor).
The corresponding filtered model for the noiseless signal is then
$P_R(x|\phi) = \frac{P(x|\phi) R(x)}{\mathcal{N}_R(\phi)}$, where $\mathcal{N}_R$
is a $\phi$-dependent normalisation factor $\mathcal{N}_R(\phi) = \sum_x R(x) P(x|\phi)$.
We can then construct a filtered explicitly-unbiased maximum-likelihood estimator like in definition~\ref{def:eumle}, by maximizing
\begin{align} \label{eq:likelihood-from-filtered-data}
    \widetilde{\mathcal{L}}_R(\phi) &= \frac{|\boldsymbol\alpha|_1}{|\mathcal{S}|} \sum_{(j_i, y_i) \in \mathcal{S}} \sgn(\alpha_{j_i}) \log P_R(y_i|\phi).
\end{align}
Regularization can be reintroduced analogously as in section~\ref{sec:rEUMLE}.
Depending on the chosen size of the {filtering window~$\Delta$}, the regularization constant $c$ can be lowered compared to the unfiltered case, or it may not be needed at all.

Let us formalize the above:
\begin{definition}\label{def:filtered-eumle}
    [Filtered EUMLE]
    Given the same assumptions as Definition \ref{def:eumle}, as well as an indicator function $R(x)$ (taking values in $\{0, 1\}$), the filtered explicitly unbiased maximum likelihood estimator (filtered EUMLE) of the phase $\phi$ is defined as $\widetilde\phi \sim \argmax_\phi \widetilde{\mathcal{L}}_R(\phi)$ where $\widetilde{\mathcal{L}}_R(\phi)$ [Eq.~\eqref{eq:likelihood-from-filtered-data}] is constructed by sampling $\widetilde{M}$ points ${\{(j_i, y_i) | j_i\sim|\alpha_{j_i}|/|\boldsymbol{\alpha}|_1, y_i \sim Q_{j_i}(y_i)\}_{i=1,...,M}}$ and constructing a subset $\mathcal{S}$ [Eq.~\eqref{eq:filtered-subset}] by rejecting all points where $R(y_i)=0$.
\end{definition}
The filtered EUMLE is also a consistent estimator, as $\widetilde{\mathcal{L}}_R(\phi)$ matches the log-likelihood of filtered noiseless data coming from the probability $P_R(x)=P(x) \cdot R(x)/\mathcal{N}_R$ [which matches the model $P_R(x|\phi)$ for the true value of $\phi$]: 
\begin{equation}
    \lim_{\widetilde{M}\to\infty} \widetilde{\mathcal{L}}_R(\phi) = \mathcal{N}_R \,
    \E\!\big[\log P(x|\phi) R(x)/\mathcal{N}_R(\phi) \big| x \sim P(x) R(x)/\mathcal{N}_R \big]
\end{equation}
\begin{theorem} The filtered EUMLE estimate $\widetilde{\phi}$ of a phase $\phi\in[\phi_{\mathrm{guess}}-\Delta,\phi_{\mathrm{guess}}+\Delta]$ is asymptotically normal with mean $\phi$ and variance $\sigma^2\propto \widetilde{M}^{-1}$.
\end{theorem}
This result can be proven with the same technique as the result without filtering in Thm.~\ref{thm:eumle}.

The scaling of the filtered EUMLE variance with fidelity depends on the details of the noise and of the filtering window.
Let us assume that $\sigma\ll \Delta$, and that $\phi$ lies in our window.
Then, in our worst-case scenario $\phi$ lies at the edge of our window, so we reject half of the data from $P(x)$, while an adversarial noise model gives samples within the window.
In the best-case scenario, all noise falls outside the window and $\phi$ lies in the center.
In this case, filtering works exactly as post-selection.
In practice, we suggest that the result may be closer to the latter than the former, as good classical bounds can be many times smaller than the spectral range of the Hamiltonian, and we have no reason to expect the outcomes due to noise to fall within the filtering window with a particularly high probability.

Filtering can also be useful when we do not have access to an exact eigenstate $\ket{\phi}$.
Assume that the starting state has nonzero overlap with multiple eigenphases separated by gap $\Delta$, and the precision of phase estimation for any single circuit $\frac{1}{\mathcal{T}}\ll\Delta$. 
Then, a filtering procedure similar to the one we described should enable to separate the signal coming from each eigenstate.
The developement of a filtering procedure that works both for eignenvalue projection and variance reduction, as well as the details on the choice of the center and width of the filter and the calculation of the filtering-EUMLE overhead are left to future work.

\section{Surface-code resource estimate and optimization}\label{sec:compilation}

In this section, we aim to benchmark our circuit-division QPE method and get an idea of the runtime overhead implied by a limitation on the number of available physical qubits.
To do this, we estimate the physical resources for running fixed-precision circuit-division sin-state QPE on qubitized block-encodings of Fermionic Hamiltonian.
In particular, we consider the 2-dimensional Fermi-Hubbard Hamiltonian \cite{babbushEncoding2018} and a selection of molecular Hamiltonians for in small to medium active spaces, using a block encoding based on the tensor hypercontraction factorization \cite{leeEven2021}.
We assume free and perfect state preparation -- including a cost for state preparation would increase the circuit division overhead as the state needs to be prepared for every run; imperfect state preparation would require filtering or a different analysis of data processing.
We also assume that the residual noise (after error correction) follows the global depolarizing model. 
This allows us to apply the results from App.~\ref{sec:sin_state_qpe_multi_circuit} and choose the optimal depth $\mathcal{T}$ and the number of samples $M$ as a function of the error rate $\gamma$ as detailed in Algorithm~\ref{alg:mle-sinqpe}.
In practice, in the presence of a more realistic noise model, the techniques proposed in Sec.~\ref{sec:unbiasing_qpe} would need to be applied, resulting in an additional overhead which depends strongly on the assumptions about the noise model.

\topic{translate to surface code resources}
In the following sections, we describe our procedure for estimating the logical qubit and Toffoli gate count of the QPE circuits for the target problems.
In order to translate these to a physical resource estimate, we must make some assumptions about the error correction and magic state distillation model used.
We consider a rotated surface code, with Clifford gates performed via lattice surgery \cite{fowlerLow2019} and CCZ2T magic state factories \cite{gidneyEfficient2019} used to produce CCZ states (as our algorithm only requires Toffoli gates as non-Clifford resources).
We consider both using a single factory, to minimize the number of physical qubits (as in \cite{babbushEncoding2018}) and using two factories to reduce the runtime.
We assume $50\%$ of routing overhead \cite{gidneyEfficient2019} on top of the physical qubits used for data encoding and magic state distillation.
We assume a physical gate error rate of $p=0.001$ and a surface code cycle time of $\SI{1}{\micro\second}$ \cite{fowlerLow2019}.
The code distances for data qubits encoding $d$ and for the two levels of distillation $d_0$ and $d_1$ are then optimized (by grid search) to minimize the physical computation volume while keeping the total error probability per circuit below $1 - e^{-\gamma \mathcal{T}}$.

\subsection{Cost estimates for the qubitized Fermi-Hubbard Hamiltonian}

We consider quantum phase estimation of the qubitized walk operator for the 2-dimensional Fermi-Hubbard model.
We consider a model with $L \times L$ sites, for a total of $2L^2$ spin-orbitals, and interaction strength $u/t = 4$ (with the hopping parameter $t$ setting the energy dimensions).
This choice of interaction strength is common in resource estimates literature, as it represents the most challenging regime to treat at scale with classical algorithms \cite{leblancSolutions2015}.
We compute the Toffoli cost of the controlled-walk operator $\mathcal{W}$ using Qualtran \cite{harriganExpressing2024} and following the procedure of \cite{babbushEncoding2018}.
We assume the rotation angles in the \textsc{PREPARE} step are encoded using $\beta = 10$ bits and performed with the phase-gradient technique (we neglect the cheap one-time preparation of the phase-gradient resource state).
We separately compute the Toffoli cost of sin-state preparation and quantum Fourier transform on $\lceil \log_2(\mathcal{T}+1) \rceil$ qubits, also using Qualtran.
We estimate the number of logical data qubits analytically as
\begin{equation}
    \overbrace{2 L^2}^{\text{spin-orbitals}}
    +
    \overbrace{\lceil \log_2(\mathcal{T}+1) \rceil}^{\text{QPE control qubits}}
    + \overbrace{\lceil 12 + 3\log_2(2 L^2) \rceil}^{\text{qubitization ancillas Eq.~64 of \cite{babbushEncoding2018}}}
    + \overbrace{\beta + 1}^{\text{controlled phase gradient}}
\end{equation}
The walk operator has the same spectrum as $\pm \arccos{H/\lambda}$, with $\lambda / t = 8 L^2$ the qubitization 1-norm.
We fix the goal of estimating the energy to a precision $\Delta E / t = 0.01$, which can be achieved by requiring a precision on the eigenphase of $\mathcal{W}$ of $\epsilon = \Delta E / \lambda = 0.01/8L^2$~\cite{babbushEncoding2018}.

\topic{full algorithm logical cost}
With the required precision on the phase estimation of the walk operator $\mathcal{W}$ fixed, we estimate the logical cost of implement the circuit-division QPE algorithm for different choices of the noise rate $\gamma$ (defined per application of $\mathcal{W}$) using the protocols developed previously in this work.
We first use Alg.~\ref{alg:mle-sinqpe} to choose the optimal $\mathcal{T}$ and $\mathcal{T}_\text{tot}$ based on $\gamma$ and $\epsilon$.
Then, multiplying these by the Toffoli cost of $\mathcal{W}$ we obtain the single-circuit Toffoli cost and total Toffoli cost, respectively.

\topic{Results and comments}
The resulting physical resources as a function of $\gamma$ are reported in Fig.~\ref{fig:fermi_hubbard_resources}, for four lattice side sizes of $L=2, 5, 10, 20$.
Increasing the error rate allowed shortens the largest circuit that is optimal to run (following Algorithm~\ref{alg:mle-sinqpe}, $\mathcal{T} \leq \gamma^{-1}$).
As we reduce $\mathcal{T}$, we must increase the number of repetitions quadratically to achieve the same fixed error rate.
Thus the total runtime scales inversely with the runtime of the largest circuit.
At the same time, increasing the allowed error rate allows us to choose smaller code distances for data ($d$) and distillation ($d_0, d_1$), reducing the physical qubit requirements.
In the right panel of Fig.~\ref{fig:fermi_hubbard_resources} we show the time required to run the algorithm under a limitation on the total physical footprint.
In the $L=5$ case, for example, we see that our algorithm requires about two hours with $3\times10^5$ physical qubits, or about $200$ hours with half the number of physical qubits.
This plot also indicates the number of physical qubits required for it to be more convenient to use multiple factories.

We observe that the curves representing the computational space-time tradeoff (in the right panel of Fig.~\ref{fig:fermi_hubbard_resources}) have an ``elbow-like'' shape.
We expect our method will lead to the most significant advantage around the elbow of the curve, as further increasing the number of available qubits affects the total runtime only marginally, while further reducing the computational time requires unreasonable space overhead.
This elbow point represents the regime in which is most advantageous to combine error correction and error mitigation, similar to the enhancement regime discussed in section~\ref{sec:combining_EM_and_qec} and represented in the bottom panel Fig.~\ref{fig:NISQ_crossover}.
We find that, across system sizes, the elbow point corresponds to the algorithm prescribing around $M\approx15$ to $40$ samples, i.e. to the middle regime in Alg.~\ref{alg:mle-sinqpe}, and to the regime where MSQPE slightly outperforms RPE (see Fig.~\ref{fig:MSQPE_vs_RPE}).

\begin{figure}
    \centering
    \includegraphics[]{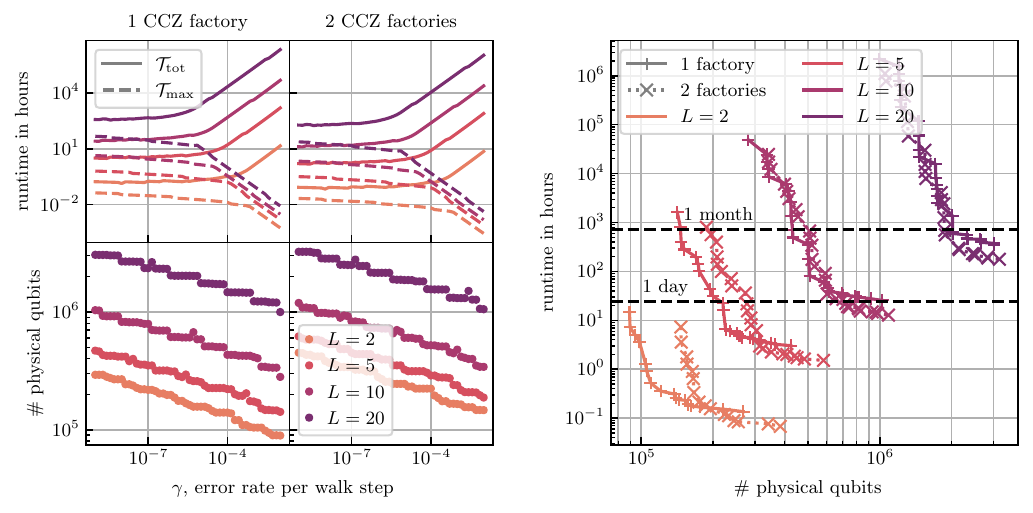}
    \caption{
        Physical costs for MSQPE applied to the qubitized Hubbard model Hamiltonian, with a target precision on the phase of $\epsilon = 10^{-2}/\lambda$ (where $\lambda$ is the qubitization 1-norm).
        (Left) runtime and qubit costs as a function of the error rate per step of the qubitized walk $\gamma$.
        (Right) trade-off between number of physical qubits and runtime.
        We assume the residual error (originating from imperfect error correction and imperfect distillation) can be modeled as a global depolarizing noise channel of strength $\gamma$.
        The logical circuits are defined according to Algorithm~\ref{alg:mle-sinqpe} (multi-circuit sin-state QPE with global depolarizing noise), with the Toffoli cost for the controlled walk operator estimated using Qualtran.
        We assume a model of fault-tolerant computation based on CCZ resource states, and show cost estimates based on using either a single CCZ factory or two of them; this allows to highlight the crossover point when it can be convenient to invest more physical qubits to implement multiple magic state factories.
        The size of surface code and CCZ factory parameters are chosen to minimize computation volume while keeping the error rate below $\gamma$, assuming physical error rate of $10^{-3}$ and surface code clock cycle of $\SI{1}{\micro\second}$.
    }
    \label{fig:fermi_hubbard_resources}
\end{figure}

\subsection{Cost estimates for Electronic Structure Calculations using THC}
\label{app:molecues_with:thc_details}
We compute the quantum resources for a number of molecular systems of small to medium size to illustrate how the higher complexity of molecular Hamiltonians and the more stringent precision requirements affect the required quantum resources and run-times.
All quantum chemistry calculations are performed with the help of PySCF~\cite{sunRecent2020}.
For \ce{H2O}, \ce{N2}, naphthalene, and anthracene the active spaces were constructed either from restricted Hartree-Fock (RHF) orbitals or with AVAS~\cite{Sayfutyarova_2017} as implemented in PySCF.
The CAS(27, 26) of the Co(salophen) complex was determined with the help of the Active Space Finder (ASF) \cite{asf_code} software developed by HQS Quantum Simulations and Covestro (see Table~\ref{tab:active_spaces} for details).

The resulting active spaces Hamiltonians were first symmetry-shifted as proposed in Ref.~\cite{Rocca_2024}.
This means finding shifts $a_1$ and $a_2$ and modifying the molecular Hamiltonian $H$ by adding powers of the total particle number operator $N_e$ according to 
\begin{equation}
    H' = H + a_1 N_e + a_2 N_e^2 .
\end{equation}
This leaves the eigenvectors and the spectrum of the Hamiltonian inside the particle number sectors intact up to a particle number dependent offset that can be easily subtracted classically but reduces the norm of the coefficients of the Hamiltonian.
Concretely, one picks $a_1$ as the median of the eigenvalues of the one-body part of the Hamiltonian and $a_2$ as the median of a certain diagonal of the two-body tensor \cite{Rocca_2024}. 

The shifted Hamiltonian $H'$ of each molecule was then factorized with the tensor hypercontraction (THC) procedure with L2 regularization with
the penalty parameter $\rho_\mathrm{THC}$ as implemented in OpenFermion~\cite{OpenFermion}, leading to a $\lambda$ factor denoted by $\lambda'$.
Symmetry shifting leaves the coupled cluster singles doubles (CCSD) energy unaffected, but does cause very small changes in the triples correction contained in CCSD(T).
It was thus more convenient to chose the THC rank based on the constraint that the CCSD energy error, $\Delta E_\mathrm{CCSD}$, stayed below 1 mEh.
Consistent with what was reported in \cite{fomichev2024initialstatepreparationquantum}, we also observed roughly a factor of two reduction of the lambda value due to symmetry shifting.

\begin{table}[ht]
    \centering
    \caption{Orbital selection method, selected active spaces, and THC factorization parameters and results for the molecular systems used in the resource estimates. The AVAS \cite{Sayfutyarova_2017} and ASF \cite{asf_code} methods were used for some systems.}\label{tab:active_spaces}
\begin{tabular}{lllrrrr}
\toprule
Molecule & Active Space & Selection Method & \hspace{-1em} THC Rank & $\rho_\mathrm{THC}$ & $\lambda'$ & $\Delta E_\mathrm{CCSD}$ [mEh] \\
\colrule
\ce{N2} & CAS(6, 6)         & around HOMO-LUMO          & 30 & $10^{-4}$ & 4.66 & 0.62 \\
\ce{H2O} & CAS(8, 6)        & AVAS (H $1s$, O $2s2p$) & 30 & $10^{-4}$ & 6.97 & 0.63 \\
naphthalene & CAS(10, 10)   & AVAS (C $2p_z$)               & 45 & $10^{-5}$  & 6.45 & 0.67 \\
anthracene & CAS(14, 14)    & AVAS (C $2p_z$)               & 80 & $10^{-5}$ & 12.12 & 0.19 \\
Co(salophen) & CAS(27, 26)  & ASF default settings          & 100 & $10^{-5}$ & 34.58 & 0.44 \\
\botrule
\end{tabular}
\end{table}

The logical resources for obtaining QPE estimates with $\epsilon = 10^{-3}/\lambda'$ precision (see Fig.~\ref{fig:thc_resources}) were then computed with Qualtran \cite{harriganExpressing2024}, from the THC tensors.
The logical circuit for the controlled walk operator is constructed using the algorithmic primitives included in Qualtran, and from that both the number of logical qubits and Toffoli gates are computed (assuming that rotations in \textsc{prepare} and \textsc{select} are implemented to 10-bit and 16-bit precision, respectively \cite{leeEven2021}).
The logical cost for the controlled walk operator is converted into physical costs for Algorithm~\ref{alg:mle-sinqpe} as a function of the residual noise rate $\gamma$, under the same assumptions made for the Hubbard model.

\begin{figure}
    \centering
    \includegraphics[]{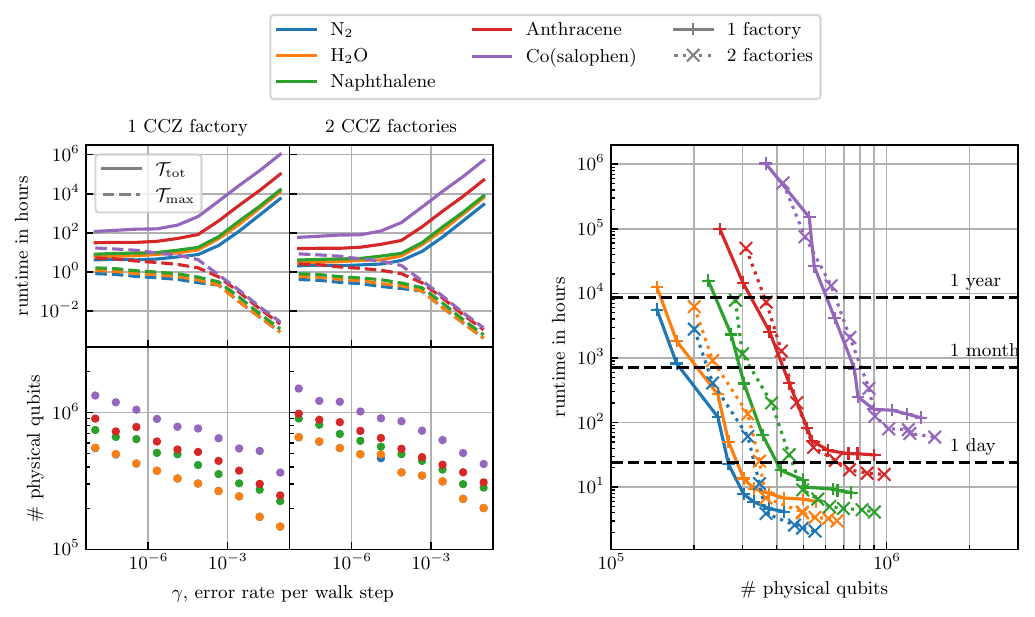}
    \caption{
        Physical costs for MSQPE applied to the electronic structure Hamiltonian factorized with THC, with a target precision on the phase of $\epsilon_\mathrm{t} = 10^{-3}/\lambda'$ (where $\lambda'$ is the symmetry-shifted THC lambda value).
        (Left) runtime and qubit costs as a function of the error rate $\gamma$ per step of the qubitized walk.
        (Right) relation between number of physical qubits and runtime.
        We assume the residual error (originating from imperfect error correction and imperfect distillation) can be modeled as a global depolarizing noise channel of strength $\gamma$.
        The logical circuits are defined according to Algorithm~\ref{alg:mle-sinqpe} (multi-circuit sin-state QPE with global depolarizing noise), with the Toffoli cost and logical qubit requirements for the walk operator estimated using Qualtran (assuming rotations in \textsc{prepare} and \textsc{select} are specified to 10 and 16 bits of precision respectively).
        We assume a model of fault-tolerant computation based on CCZ resource states, and show cost estimates based on using either a single CCZ factory or two of them; this allows to highlight the crossover point when it can be convenient to invest more physical qubits to implement multiple magic state factories.
        The size of surface code and CCZ factory parameters are chosen to minimize computation volume while keeping the error rate below $\gamma$, assuming physical error rate of $10^{-3}$ and surface code clock cycle of $\SI{1}{\micro\second}$.
    }
    \label{fig:thc_resources}
\end{figure}

\subsection{Comparison to the fully fault-tolerant algorithm}
\label{sec:ftqc_comparison}
For the sake of comparison in Fig.~\ref{fig:ftqc-res-comparison} we also estimated the resources required for running the standard fault-tolerant single-circuit sin-state QPE algorithm.
The standard fault tolerant algorithm is defined by choosing an acceptable failure probability $\delta$ as a free parameter.
We fix the logical circuit to depth $\mathcal{T} = \mathcal{T}_\text{tot} = \lambda \pi/\Delta E$ to ensure a phase estimate to the desired precision in the absence of errors.
The surface code and distillation distances are then chosen to optimize computational volume while keeping the total error probability below $\delta$, under the same assumptions as above.
The resulting requirements are reported in Fig.~\ref{fig:ftqc-res-comparison} for $\delta = 1\%$ ($\blacktriangle$) and $\delta = 0.1\%$ ($\blacktriangledown$), together with the requirements for the circuit-division algorithm.

The costs of this standard approach should only be considered a reference and not a fair comparison, as the circuit-division algorithms targets the estimate of the phase with variance bounded by $\epsilon^2$, without allowing for any failure probability.
A failure probability of $\delta$, assuming the error produces a random phase estimate, would result in a final estimate with a much larger 
variance $(1-\delta) \epsilon^2 + \delta \lambda^2 \pi^2$.
While prior knowledge about the target phase can be used to filter these results with a repeat-until-success strategy, effectively reducing the expected error, the same technique can be applied to circuit-division QPE as suggested in App.~\ref{sec:filtering}.

Note that, in most cases, the cost of phase estimation circuits on a qubitized walk operator can be improved by an additional factor of $\approx 2$ by using the unary iteration circuit to control the reflect operator as described in \cite{leeEven2021}.
This improvement is not considered in this work, as it would require a subtle modification of the parameter selection procedure for MSQPE.
Nevertheless, we expect this would yield an improvement by about a factor 2 in total runtime.

\bibliography{bibliography-zotero}

\end{document}